\pdfoutput=1
\documentclass[11pt]{article}

\usepackage[margin=1in]{geometry}
\usepackage[T1]{fontenc}
\usepackage[utf8]{inputenc}
\usepackage{lmodern}
\usepackage{microtype}
\usepackage{amsmath,amssymb,amsfonts,amsthm,mathtools}
\usepackage{bm}
\usepackage{booktabs}
\usepackage{array}
\usepackage{enumitem}
\usepackage[section]{placeins}
\usepackage{float}
\newfloat{algorithm}{tbp}{loa}
\floatname{algorithm}{Algorithm}
\usepackage{hyperref}

\hypersetup{
  colorlinks=true,
  linkcolor=blue,
  citecolor=blue,
  urlcolor=blue,
  pdftitle={Fixed-Composition Shuffle Asymptotics in the Full-Support Gaussian Regime},
  pdfauthor={Alex Shvets}
}

\newtheorem{theorem}{Theorem}[section]
\newtheorem{proposition}[theorem]{Proposition}
\newtheorem{lemma}[theorem]{Lemma}
\newtheorem{corollary}[theorem]{Corollary}
\newtheorem{definition}[theorem]{Definition}
\newtheorem{remark}[theorem]{Remark}

\usepackage[nameinlink,capitalize]{cleveref}

\crefname{theorem}{Theorem}{Theorems}
\Crefname{theorem}{Theorem}{Theorems}
\crefname{proposition}{Proposition}{Propositions}
\Crefname{proposition}{Proposition}{Propositions}
\crefname{lemma}{Lemma}{Lemmas}
\Crefname{lemma}{Lemma}{Lemmas}
\crefname{corollary}{Corollary}{Corollaries}
\Crefname{corollary}{Corollary}{Corollaries}
\crefname{remark}{Remark}{Remarks}
\Crefname{remark}{Remark}{Remarks}
\crefname{definition}{Definition}{Definitions}
\Crefname{definition}{Definition}{Definitions}

\newcommand{\R}{\mathbb{R}}

\newcommand{\Z}{\mathbb{Z}}
\newcommand{\E}{\mathbb{E}}
\newcommand{\Pp}{\mathbb{P}}
\newcommand{\Var}{\operatorname{Var}}
\newcommand{\Cov}{\operatorname{Cov}}
\newcommand{\diag}{\operatorname{diag}}
\newcommand{\KL}{\operatorname{KL}}
\newcommand{\JSD}{\operatorname{JSD}}

\newcommand{\supp}{\operatorname{supp}}
\newcommand{\one}{\mathbf{1}}
\newcommand{\eps}{\varepsilon}
\newcommand{\dd}{\mathrm{d}}
\newcommand{\Y}{\mathcal{Y}}
\newcommand{\T}{\mathsf{T}}

\newcommand{\PhiG}{\Phi}
\newcommand{\phiG}{\varphi}

\newcolumntype{L}[1]{>{\raggedright\arraybackslash}p{#1}}
\title{Fixed-Composition Shuffle Asymptotics in the Full-Support Gaussian Regime\thanks{A preliminary version of this work appeared on Zenodo (DOI: 10.5281/zenodo.18112202). This manuscript is the journal revision of arXiv:2602.09029 and Part~I of the series described below.}}
\author{Alex Shvets\\Independent Researcher, Haifa, Israel\\\texttt{alex@shvets.io}; \texttt{alt178332@gmail.com}\\ORCID: 0009-0005-9802-379X}
\date{Revised manuscript, July 2026}

\begin{document}
\maketitle

\begin{abstract}
We study privacy amplification by shuffling for binary-input local randomizers with a fixed finite output alphabet and full support.  For a dataset containing exactly $k$ ones among $n$ users, let $T_{n,k}$ denote the shuffled histogram law.  For fixed-composition neighboring shuffled histogram laws in the interior regime, we identify the covariance and Fisher constant governing the neighboring pair $(T_{n,k},T_{n,k+1})$.  For a composition parameter $\pi\in[0,1]$ the correct covariance is $\Sigma_\pi=(1-\pi)\Sigma_0+\pi\Sigma_1$, where $\Sigma_b=\diag(W_b)-W_bW_b^\top$, rather than the multinomial covariance of the mixture $(1-\pi)W_0+\pi W_1$.  With $v=W_1-W_0$, the resulting constant is $I_\pi=v^\top\Sigma_\pi^+v$ and is strictly larger than the mixture-covariance proxy for every nontrivial channel and every interior composition.  We prove exact likelihood-ratio identities and the regression decomposition $L_{n,k}-1=n^{-1}s_n^\top(N-\mathbb EN)+R_{n,k}$, where the score $s_n$ is evaluated at the empirical composition, with $\mathbb E R_{n,k}^2=O(n^{-2})$ and $\mathbb E R_{n,k}^4=O(n^{-4})$, uniformly over interior compositions.  Consequently, for these fixed-composition neighboring laws, the Jensen--Shannon divergence equals $I_{k/n}/(8n)+O(n^{-2})$, and the same constant governs smooth divergence asymptotics.  For $\mu_n=\sqrt{I_{k/n}/n}$, both directed hockey-stick privacy curves at $\varepsilon=t\mu_n$ equal $\mu_n\{\varphi(t)-t\Phi(-t)\}+O(n^{-1})$, uniformly for $t$ in compact sets.  This local asymptotic statement is separated from finite-sample certification, for which we give exact positive-part accounting formulas.  Binary-output and fixed-message unbundled specializations, including an exact variance decomposition, are also provided.
\end{abstract}


\tableofcontents

\section{Introduction}

The shuffle model of differential privacy inserts a trusted shuffler between users and an analyst.  Each user applies a local randomizer and sends a message; the shuffler releases the messages in random order.  For a binary-input local randomizer $W:\{0,1\}\to\Delta(\Y)$ with finite output alphabet $\Y=\{1,\ldots,d\}$, the shuffled output is equivalently the histogram
\[
        N=(N_y)_{y\in\Y},\qquad N_y=\sum_{i=1}^n {\bf 1}\{Y_i=y\}.
\]
If the input dataset has exactly $k$ ones, the histogram law is denoted $T_{n,k}$.  By exchangeability, every neighboring binary pair of datasets induces some neighboring pair $(T_{n,k},T_{n,k+1})$.

This paper studies the asymptotic structure of such neighboring experiments when the alphabet and the local randomizer are fixed, the randomizer has full support, and $k/n$ remains in a compact subinterval of $(0,1)$.  This is the full-support Gaussian regime.  The goal is not to replace finite-$n$ privacy accounting algorithms.  Those algorithms answer the certification question: for this mechanism and this $n$, what valid $(\eps,\delta)$ guarantee can be reported?  The goal here is structural: to identify the correct Gaussian tangent scale and score geometry, the correct covariance, and the constants controlling local privacy and information leakage for fixed-composition neighboring histogram pairs.

\subsection{Scope and limitations}

The main results, in \cref{sec:covariance,sec:lr,sec:linearization,sec:jsd,sec:local-privacy,sec:finite-accounting,sec:unbundled}, are intentionally limited to the following regime: binary inputs, fixed finite output alphabet, fixed local randomizer, full support, and neighboring datasets differing in one user.  They do not cover growing alphabets, sparse supports, $n$-dependent alphabets, or non-absolutely-continuous pairs without a separate singular-component analysis.  \Cref{sec:boundary} separately treats one canonical randomized-response boundary regime in which the local parameter $\eps_0(n)$ is allowed to grow subject to $a_n=e^{\eps_0(n)}/n\to0$; no general theorem for randomizers with an $n$-dependent privacy parameter is claimed.  The results also do not assert that a Gaussian differential privacy (GDP) approximation is by itself a valid finite-$n$ privacy certificate.

This distinction is essential.  In the Gaussian regime the GDP parameter is
\[
        \mu_n(\pi)=\sqrt{I_\pi/n}=O(n^{-1/2}).
\]
Therefore a global trade-off error of order $O(n^{-1/2})$ is of the same order as the Gaussian signal.  Such a statement can identify the tangent scale, but it is too coarse to certify a small fixed-$\eps$ privacy curve.  The refined privacy statement in this revision is instead formulated on the natural local scale $\eps=t\mu_n$, where each directed privacy curve has size $\mu_n$ and an $O(n^{-1})$ remainder is genuinely smaller than the main term.  For finite-$n$ certification, one should use the exact positive-part formulas in \cref{sec:finite-accounting} or another explicit computable upper bound.

This manuscript is the journal revision of arXiv:2602.09029.  It narrows the
earlier claims: global GDP-style, Le Cam-equivalence, and finite-$n$ dominance
statements are replaced by a two-sided local privacy-curve theorem and by an
exact accounting layer.  It is Part~I of a series.  Part~II
\cite{ShvetsPartII2026} treats the critical Poisson, Skellam, and
compound-Poisson regimes, while Part~III \cite{ShvetsPartIII2026} develops the
dominant-block quotient geometry and hybrid Gaussian--compound-Poisson limits.
Separate companion work treats canonical growing-alphabet experiments
\cite{ShvetsGrowing2026} and anchored likelihood-ratio geometry and exact
privacy envelopes \cite{ShvetsAnchored2026}.  Those papers are logically
independent of the proofs in the present fixed-alphabet, full-support regime.

Fixing a neighboring composition pair does not turn the analysis into a worst-case finite-$n$ theorem.  By exchangeability every neighboring binary pair induces some $(T_{n,k},T_{n,k+1})$.  \Cref{prop:boundary-worstcase} proves only that the continuum tangent profile $\pi\mapsto I_\pi$ is convex and therefore has its largest endpoint value at $\pi=0$ or $\pi=1$.  The two endpoint constants correspond to the canonical experiment in the two possible channel orientations.  This does not identify the exact maximizing integer composition at finite $n$, because the interior expansion is not uniform in the boundary layer and lower-order terms may change the maximizer.  The exact worst-case curve remains the finite maximum in \eqref{eq:worst-case-dp} and must be evaluated or bounded through the accounting layer of \cref{sec:finite-accounting}.

\subsection{Main contributions}

We separate the supporting framework from the new results, and we do not claim
the framework as a contribution.

\paragraph{Supporting framework.}
Two ingredients are used throughout and are not new.  The first is the exact
likelihood-ratio representation.  For the canonical pair $(T_{n,0},T_{n,1})$
the likelihood ratio is linear in the histogram; for a general neighboring pair
$(T_{n,k},T_{n,k+1})$ it is a conditional expectation,
\[
        L_{n,k}(N)=\frac{\dd T_{n,k+1}}{\dd T_{n,k}}(N)
        =1+\E[r(Y_\star)\mid N],
        \qquad r(y)=\frac{W_1(y)}{W_0(y)}-1 .
\]
The canonical case appears in equivalent form in Girgis et al.
\cite{GirgisEtAl2021}; we record the general fixed-composition version because
it is the entry point for everything below.  The second ingredient is the
positive-part accounting identity of \cref{sec:finite-accounting}.  That
identity is standard; what we supply is its specialization to shuffled
histogram experiments and a clean separation between the exact accounting layer
and the asymptotic layer.  No new accounting principle is claimed.

\paragraph{Geometric identification.}
The tangent geometry of a fixed-composition neighboring pair is not the i.i.d.
mixture geometry.  If $k/n\to\pi$, the histogram is not a multinomial sample
from $(1-\pi)W_0+\pi W_1$; it is a fixed-composition sum of exactly $n-k$ draws
from $W_0$ and $k$ draws from $W_1$, so the covariance is
\[
        \Sigma_\pi=(1-\pi)\Sigma_0+\pi\Sigma_1,
\]
not the mixture covariance $\diag(f_\pi)-f_\pi f_\pi^\top$ with
$f_\pi=(1-\pi)W_0+\pi W_1$.  The identity itself is elementary.  Its
consequence is not: the two covariances differ by the rank-one term
$\pi(1-\pi)vv^\top$, and \cref{prop:fisher-proxy} shows that the mixture proxy
gives a strictly and always optimistically small Fisher constant, $I_\pi>I_f$,
for every $d\ge2$, every nontrivial channel, and every interior composition.
The effect is quantitatively significant: for the four compositions displayed
for the three-symbol channel in \cref{subsec:numerics}, the mixture constant is
smaller by $20.0$--$31.0\%$.  This is a warning against a natural but wrong
i.i.d.-mixture heuristic; it is not meant to attribute that substitution to any
particular prior shuffle-DP proof.  Composition dependence is a separate issue:
it is absent for variance-symmetric binary channels, can occur for
variance-asymmetric binary channels, and can also occur for larger alphabets;
it is not implied solely by the condition $d\ge3$.

\paragraph{New results.}
Three results are new.

The first is a uniform residual decomposition for the posterior likelihood
ratio.  Uniformly for $k/n\in[\eta,1-\eta]$, one can write under $P=T_{n,k}$
\[
        L_{n,k}-1=U_{n,k}=\frac{1}{n}Z_{n,k}+R_{n,k},
        \qquad Z_{n,k}=s_n^\top(N-\E_PN),
\]
with $s_n$ the score at the empirical composition, $\E_P[R_{n,k}Z_{n,k}]=0$,
$\E_P R_{n,k}^2=O(n^{-2})$, and $\E_P R_{n,k}^4=O(n^{-4})$.  This is strictly
stronger than a pointwise typical-set linearization: the bounds are $L^2$ and
$L^4$ bounds valid uniformly over interior compositions, not statements on a
high-probability event.  It is the load-bearing input for both results below.

The second is the exact information constant.  In the interior regime,
\[
        \JSD(T_{n,k}\Vert T_{n,k+1})=\frac{I_{k/n}}{8n}+O(n^{-2}),
        \qquad I_\pi=v^\top\Sigma_\pi^+v,
\]
and the same constant governs smooth $f$-divergences (\cref{thm:f-divergence}).
The remainder is $O(n^{-2})$, one full order below the leading term.

The third is a two-sided local privacy curve.  For
$\mu_n=\sqrt{I_{k/n}/n}$ and $t$ in compact subsets of $[0,\infty)$,
\[
        \delta_{T_{n,k+1}\Vert T_{n,k}}(t\mu_n)
        =\mu_n\{\phiG(t)-t\PhiG(-t)\}+O(n^{-1}).
\]
The leading term has size $\Theta(n^{-1/2})$ and the remainder is smaller by a
factor $n^{-1/2}$, so the statement is sharp on its own scale.  The same
Gaussian curve governs the reverse ordering (\cref{cor:two-sided}), which makes
it a genuine two-sided $(\eps,\delta)$ statement for interior compositions.  It
is proved by direct change of measure and a Lipschitz estimate, not by a Le Cam
argument.

\subsection{Main theorem dependency map}\label{subsec:main-map}

For readability and auditability, the logical dependencies of the main results
are as follows.  Exact likelihood-ratio identities are proved first and do not
use asymptotics.  The covariance correction then identifies the correct tangent
metric.  The lattice ratio lemmas in \cref{app:edgeworth} enter twice:
\cref{lem:ratio} supports the pointwise conditional-expectation linearization
of \cref{sec:linearization}, and its refinement \cref{lem:refined-ratio}
supports the residual bounds in \cref{prop:residual-bounds}.  Those residual
bounds are the load-bearing input for the $O(n^{-2})$ Jensen--Shannon
expansion.  The $O(n^{-1})$ local privacy-curve remainder uses both the
residual bounds and the non-uniform Berry--Esseen call-payoff estimate in
\cref{lem:call-payoff}.  The normalized privacy-loss expansion is a parallel
interpretive consequence and is not used as a black box in the privacy-curve
proof.

\begin{table}[H]
\centering
\small
\begin{tabular}{@{}L{0.22\textwidth}L{0.31\textwidth}L{0.34\textwidth}@{}}
\toprule
Layer & Main result & Output used later \\
\midrule
Exact shuffle algebra & \Cref{lem:canonical-lr,lem:general-lr} & $Q\ll P$ and $L_{n,k}=1+\E[r(Y_\star)\mid N]$ on $\supp(T_{n,k})$, together with exact accounting that retains the reverse singular mass. \\
Fixed-composition geometry & \Cref{prop:mixture-correction,prop:fisher-proxy,prop:boundary-worstcase,lem:sigma-lipschitz} & Correct covariance $\Sigma_\pi=(1-\pi)\Sigma_0+\pi\Sigma_1$ and Fisher constant $I_\pi=v^\top\Sigma_\pi^+v$. \\
Local lattice analysis & \Cref{lem:tilt-control,lem:tilted-aperiodicity,lem:edgeworth,lem:refined-ratio} & Uniform $C^1$ control of the relative Edgeworth remainder and the bounded-shift estimate used in the residual theorem. \\
Residual machinery & \Cref{lem:hist-moments,lem:regression,prop:residual-bounds}, using \cref{lem:refined-ratio} & $L_{n,k}-1=Z_{n,k}/n+R_{n,k}$ with $\E R_{n,k}^2=O(n^{-2})$, $\E R_{n,k}^4=O(n^{-4})$, and $\E[R_{n,k}Z_{n,k}]=0$. \\
Information constants & \Cref{thm:interior-jsd,thm:f-divergence} & $\JSD(T_{n,k}\Vert T_{n,k+1})=I_{k/n}/(8n)+O(n^{-2})$ and the same leading constant for smooth divergences. \\
Gaussian tangent & \Cref{thm:LAN-normalized} & A normalized privacy-loss expansion derived from the residual bounds; it is interpretive and not an input to the privacy-curve theorem. \\
Local privacy curve & \Cref{lem:call-payoff,thm:local-delta,cor:two-sided}, using \cref{prop:residual-bounds} & For $\eps=t\sqrt{I_{k/n}/n}$, both directed curves equal $\sqrt{I_{k/n}/n}\{\phiG(t)-t\PhiG(-t)\}+O(n^{-1})$ uniformly for bounded $t$. \\
Finite-$n$ certification & \Cref{prop:positive-part,cor:general-accounting,prop:chernoff} & Exact accounting identities and computable upper bounds; no Gaussian approximation is used as a certificate without an explicit error term. \\
\bottomrule
\end{tabular}
\caption{Dependency map for the main results.  The asymptotic constants and local privacy curve depend on the residual theorem, while finite-sample certification stays in the exact accounting layer.}
\label{tab:dependency-map}
\end{table}

\subsection{Relation to prior work}

Differential privacy was introduced by Dwork, McSherry, Nissim, and Smith \cite{DworkMcSherryNissimSmith2006}; see Dwork and Roth \cite{DworkRoth2014} for a textbook treatment.  The shuffle model and privacy amplification by shuffling were developed in Erlingsson et al. \cite{ErlingssonEtAl2019}, Cheu et al. \cite{CheuEtAl2019}, Balle et al. \cite{BalleBellGasconNissim2019}, and Feldman--McMillan--Talwar \cite{FeldmanMcMillanTalwar2021,FeldmanMcMillanTalwar2023}, among many others.  Further amplification-bound work includes variation-ratio reduction \cite{WangEtAl2024Unified}, decomposition-based optimal bounds \cite{SuChengWang2025Decomposition}, R\'enyi-DP and hypothesis-testing bounds \cite{ChenCaoGe2024RDP}, and a generalized framework for personalized local-privacy specifications \cite{ChenCaoGe2024Generalized}.  The latter works address amplification guarantees under broader local-privacy descriptions rather than the exact pair-specific fixed-composition constant studied here.

The distinction between fixed-composition sums and i.i.d. mixtures also belongs to the broader theory of permutation mixtures.  Han--Niles-Weed \cite{HanNilesWeed2024} bound statistical distances between high-dimensional permutation mixtures and their i.i.d. counterparts and derive, among other applications, a shuffle-privacy guarantee for Gaussian noise.  Their results give global approximate-independence bounds rather than the local categorical neighboring coefficient studied here; in particular, they do not identify the fixed-composition Fisher constant $I_\pi$ or the local hockey-stick expansion.

Takagi--Liew \cite{TakagiLiew2026} develop asymptotic blanket-divergence bands
beyond pure local DP.  The two analyses use opposite scalar conventions.  Their
blanket divergence
$\mathcal D^{\rm blanket}_{e^{\eps},n,\mathcal R_{\rm ref},\gamma}$ carries a
reference law $\mathcal R_{\rm ref}$ and a mass parameter $\gamma$, and the
associated shuffle index is
\[
        \chi_{\rm TL}=\frac{\sqrt\gamma}{\sigma},
        \qquad
        \sigma^2=\Var_{Y\sim\mathcal R_{\rm ref}}\bigl(l_0(Y)\bigr),
        \qquad
        l_0=\frac{\mathcal R_{x_1}-\mathcal R_{x_1'}}{\mathcal R_{\rm ref}} .
\]
Here $\gamma$ is the blanket mass when $\mathcal R_{\rm ref}$ is the blanket
distribution, so $\gamma<1$ for every nontrivial channel.  The value $\gamma=1$
arises instead in the lower-bound configuration
$\mathcal D^{\rm blanket}_{e^{\eps},n,\mathcal R_x,1}$ inherited from
Su--Cheng--Wang \cite{SuChengWang2025Decomposition}.  Taking the canonical
reference $\mathcal R_{\rm ref}=W_0$ in that configuration gives
\[
        \sigma^2
        =\Var_{W_0}\!\left(\frac{W_1-W_0}{W_0}\right)
        =\chi^2(W_1\Vert W_0)=I_0,
\]
so their index reduces to $\chi_{\rm TL}=I_0^{-1/2}$: larger $\chi_{\rm TL}$
means stronger privacy, whereas larger $I_0$ means a larger local privacy
signal.  The two regimes are moreover disjoint.  Their asymptotics require
$\eps_n=\omega(n^{-1/2})$ together with $\eps_n=O(\sqrt{\log n/n})$, a
moderate-deviation band that excludes the local scale
$\eps=t\sqrt{I_{k/n}/n}=\Theta(n^{-1/2})$ analyzed here.  Their subsequent
mechanism-design work uses the shuffle-index framework for private vector mean
estimation \cite{TakagiLiewOptimization2026}.

GDP was introduced by Dong--Roth--Su \cite{DongRothSu2022}.  Exact and numerical accounting methods include general FFT and characteristic-function approaches \cite{KoskelaJalkoHonkela2020,ZhuWang2021}, shuffle-specific privacy-loss accounting \cite{KoskelaHeikkilaHonkela2023}, and recent decomposition methods \cite{SuChengWang2025Decomposition}.  Mutual-information leakage in the shuffle model is studied by Su--Cheng--Wang \cite{SuChengWang2025MI}.  Their shuffle-DP functional $I(X_1;Z\mid X_{-1})$ specializes, for a uniform binary target input and a fixed background containing $k$ ones, to
\[
        I(X_1;Z\mid X_{-1}=x_{-1})
        =I(X_1;N\mid X_{-1}=x_{-1})
        =\JSD(T_{n,k}\Vert T_{n,k+1}),
\]
because the shuffled sequence $Z$ and its histogram $N$ are statistically equivalent.  Thus the information functional overlaps exactly in this specialization.  Su--Cheng--Wang give general local-DP-based information bounds and analyze homogeneous and heterogeneous background models; the distinct contribution here is the exact channel- and composition-specific coefficient $I_\pi/(8n)$ together with its connection to both directed local hockey-stick curves.  Quantitative-information-flow leakage for shuffled mechanisms is studied by Jurado et al. \cite{JuradoEtAl2023}.  Tight specialized accounting for shuffled randomized response is studied by Biswas--Jung--Palamidessi \cite{BiswasJungPalamidessi2024}.  The exact linear likelihood-ratio identity for the canonical shuffle pair appears in equivalent form in Girgis et al. \cite{GirgisEtAl2021}.

The comparison below is by objective, not by a claimed numerical dominance over existing amplification bounds.  Most shuffle-amplification results ask for valid worst-case central-DP, RDP, or computable accounting guarantees for a mechanism.  Here we fix a neighboring histogram pair in the interior fixed-composition regime and identify the local tangent constant $I_\pi$ governing JSD, smooth $f$-divergences, and both directed local privacy curves.  This pair-specific asymptotic descriptor is not a replacement for worst-case amplification certification; exact finite-$n$ certification remains in the positive-part/accounting layer.

\begin{table}[!t]
\centering
\footnotesize
\begin{tabular}{@{}L{0.18\textwidth}L{0.22\textwidth}L{0.22\textwidth}L{0.26\textwidth}@{}}
\toprule
Work & Model / method & Type of guarantee & Relation to this paper \\
\midrule
Balle et al. \cite{BalleBellGasconNissim2019} & Privacy blanket for shuffle amplification & General central-DP amplification upper bounds & Certification objective; not intended to identify the pair-specific fixed-composition Fisher constant $I_\pi$. \\
Feldman--McMillan--Talwar \cite{FeldmanMcMillanTalwar2021,FeldmanMcMillanTalwar2023} & Clone-based/RDP amplification analysis & Worst-case central-DP and RDP amplification bounds & Same shuffle-amplification setting at the level of guarantees, but a different functional from local tangent Fisher geometry. \\
Han--Niles-Weed \cite{HanNilesWeed2024} & Permutation mixtures versus i.i.d. counterparts & Statistical-distance bounds and Gaussian shuffle application & Global approximate-independence theory; this paper identifies the exact categorical neighboring tangent constant and local curve. \\
Takagi--Liew \cite{TakagiLiew2026} & Blanket-divergence asymptotics and shuffle index beyond pure local DP & Tight asymptotic bands and finite-$n$ FFT computation for blanket divergence & In the lower-bound configuration with reference $W_0$ their shuffle index is $\chi_{\rm TL}=I_0^{-1/2}$; their moderate-deviation band $\eps_n=\omega(n^{-1/2})$ excludes the local scale studied here. \\
Dong--Roth--Su \cite{DongRothSu2022} & GDP / $f$-DP framework & Gaussian trade-off calculus for experiments & Provides the Gaussian trade-off formalism; this paper computes the shuffle-specific tangent parameter $\mu_n=\sqrt{I_\pi/n}$. \\
Su--Cheng--Wang \cite{SuChengWang2025Decomposition} & Decomposition and FFT computation & Computable shuffle-amplification bounds & Directly relevant to finite-$n$ certification; complementary to the asymptotic identification of $I_\pi$. \\
Su--Cheng--Wang \cite{SuChengWang2025MI} & Mutual information in the shuffle model & Target-input and position leakage & Their target-input mutual information specializes to JSD for a uniform binary neighbor bit; this paper gives its exact fixed-composition coefficient and links it to local hockey-stick curves. \\
Biswas--Jung--Palamidessi \cite{BiswasJungPalamidessi2024} & Shuffled $k$-randomized response & Tight specialized privacy guarantees & Overlaps with binary-output/RR accounting; our Section~\ref{sec:finite-accounting} is an accounting specialization, not a new positive-part principle. \\
\bottomrule
\end{tabular}
\caption{Positioning by objective rather than by a uniform numerical improvement claim.}
\label{tab:related}
\end{table}

We do not study Jensen--Shannon contraction coefficients or strong data-processing constants.  In particular, the binary-input reduction for SDPI constants in Ordentlich--Polyanskiy \cite{OrdentlichPolyanskiy2022} concerns an optimization over input subchannels of a fixed channel.  Our JSD results instead evaluate the divergence between the two fixed neighboring shuffled histogram laws $T_{n,k}$ and $T_{n,k+1}$.

\subsection{Road map}

\Cref{sec:model} sets up the model, assumptions, likelihood-ratio notation, and worst-case DP curves.  \Cref{sec:covariance} proves the fixed-composition covariance formulas.  \Cref{sec:lr} gives exact likelihood-ratio identities.  \Cref{sec:linearization} proves the conditional-expectation linearization and residual decomposition.  \Cref{sec:jsd} derives Jensen--Shannon and smooth $f$-divergence asymptotics.  \Cref{sec:local-privacy} proves the Gaussian tangent and local privacy-curve results.  \Cref{sec:finite-accounting} gives exact finite-$n$ accounting identities.  \Cref{sec:unbundled} treats fixed-$m$ unbundled shuffling.  \Cref{sec:boundary} discusses the randomized-response boundary.  \Cref{sec:discussion} concludes.

\section{Model, definitions, and privacy curves}\label{sec:model}

\subsection{Shuffle histogram experiments}

Fix a finite output alphabet $\Y=\{1,\ldots,d\}$ with $d\ge2$.
The case $d=1$ is degenerate: necessarily $W_0=W_1$, all neighboring
histogram laws coincide, and all divergences and privacy curves considered
below are zero.  We therefore exclude this case.  A local randomizer is a channel
\[
        W:\{0,1\}\to\Delta(\Y),\qquad x\mapsto W_x.
\]
Given a dataset $D=(x_1,\ldots,x_n)\in\{0,1\}^n$, users release independent messages $Y_i\sim W_{x_i}$.  The shuffler reveals the histogram
\[
        N=(N_y)_{y\in\Y}\in\Z_{\ge0}^d,
        \qquad N_y=\sum_{i=1}^n {\bf 1}\{Y_i=y\},
        \qquad \sum_y N_y=n.
\]
For $k\in\{0,\ldots,n\}$, let $T_{n,k}$ denote the law of $N$ when the dataset contains exactly $k$ ones.

\begin{remark}[Composition reduction]
Because the shuffled mechanism is exchangeable in the users, $T_{n,k}$ depends only on the number $k$ of ones.  Every neighboring binary pair of datasets induces $(T_{n,k},T_{n,k+1})$ for some $k\in\{0,\ldots,n-1\}$.
\end{remark}

\subsection{Likelihood ratios and privacy curves}

For distributions $P,Q$ on a finite space with $Q\ll P$, define
\[
        L(\omega)=\frac{Q(\omega)}{P(\omega)},
        \qquad \omega\in\supp(P),
\]
and fix the representative $L(\omega)=0$ for
$\omega\notin\supp(P)$.  Thus $L=\dd Q/\dd P$ is used pointwise only
on $\supp(P)$.  On events where $L>0$, write $\Lambda=\log L$ for the
privacy loss.  All results using a bounded real-valued privacy loss impose
mutual absolute continuity through the full-support assumption below.
The one-sided hockey-stick privacy curve is
\[
        \delta_{Q\Vert P}(\eps)
        :=\sup_A\{Q(A)-e^\eps P(A)\}
        =\E_P[(L-e^\eps)_+],\qquad \eps\ge0.
\]
The two-sided curve is
\[
        \delta_{\{P,Q\}}(\eps)=\max\{\delta_{Q\Vert P}(\eps),\delta_{P\Vert Q}(\eps)\}.
\]
For the full shuffled binary mechanism, the worst-case curve is
\begin{equation}\label{eq:worst-case-dp}
        \delta_{\rm shuf}(\eps)
        =\max_{0\le k\le n-1}\max\{\delta_{T_{n,k+1}\Vert T_{n,k}}(\eps),
        \delta_{T_{n,k}\Vert T_{n,k+1}}(\eps)\}.
\end{equation}
Thus an interior asymptotic theorem for $k/n\in[\eta,1-\eta]$ is not by itself a worst-case DP theorem; boundary values of $k$ must be controlled separately or included through exact finite-$n$ accounting.

We use
\[
        \JSD(P\Vert Q)=\frac12\KL(P\Vert M)+\frac12\KL(Q\Vert M),\qquad M=\frac12(P+Q),
\]
and $f$-divergences $D_f(Q\Vert P)=\E_P[f(L)]$.

\subsection{Support assumptions}

\begin{definition}[Minimum masses]
Define
\[
        \delta_\star=\min_{y\in\Y}W_0(y),\qquad
        \delta_{\rm full}=\min_{y\in\Y}\min\{W_0(y),W_1(y)\}.
\]
\end{definition}

The identities in \cref{sec:lr} require $\delta_\star>0$ in the direction
$T_{n,k+1}\ll T_{n,k}$.  This one-sided assumption does not imply the
reverse domination.  Accordingly, the pointwise ratio in
\eqref{eq:general-lr-sum} is asserted only on $\supp(T_{n,k})$; outside
this support its numerator and denominator both vanish, and the
representative fixed above sets the likelihood ratio equal to zero.  The
exact accounting formulas in \cref{sec:finite-accounting} remain valid
under this one-sided domination and retain the singular mass in the reverse
privacy curve.  Results involving bounded privacy loss, Edgeworth ratios
uniformly over symbols, and local Gaussian privacy curves assume
$\delta_{\rm full}>0$.  Those asymptotic results do not treat the singular
regime.

\section{Fixed-composition covariance and Fisher constant}\label{sec:covariance}

Write $W_0,W_1$ as vectors in $\R^d$ and set
\[
        v=W_1-W_0,
        \qquad \one^\top v=0,
        \qquad \Sigma_b=\diag(W_b)-W_bW_b^\top,
        \quad b\in\{0,1\}.
\]
For $\pi\in[0,1]$ define the fixed-composition covariance
\begin{equation}\label{eq:sigma-pi}
        \Sigma_\pi=(1-\pi)\Sigma_0+\pi\Sigma_1.
\end{equation}
For a dataset with $k=n\pi$ ones (with the evident interpretation when $\pi=k/n$), independence of the one-hot messages gives the exact identities
\[
        \E N=(n-k)W_0+kW_1,
        \qquad
        \Cov(N)=(n-k)\Sigma_0+k\Sigma_1=n\Sigma_\pi.
\]
Let
\[
        \mathcal T=\{x\in\R^d:\one^\top x=0\}
\]
be the tangent space of the simplex.  Under full support, $\Sigma_\pi$ is positive definite on $\mathcal T$ for $\pi\in[0,1]$.  Indeed, for $x\in\mathcal T$,
\[
        x^\top\Sigma_bx=\Var_{Y\sim W_b}\{x(Y)\}.
\]
Full support makes this variance zero only when all coordinates of $x$ are equal; the tangent constraint then forces $x=0$.  The same conclusion holds for every convex combination $\Sigma_\pi$.  Let $\Sigma_\pi^+$ be its Moore--Penrose inverse, equivalently the inverse on $\mathcal T$ and zero on $\operatorname{span}\{\one\}$.  Define
\begin{equation}\label{eq:Ipi}
        I_\pi=v^\top\Sigma_\pi^+v.
\end{equation}

\begin{proposition}[Mixture covariance correction]\label[proposition]{prop:mixture-correction}
Let $f_\pi=(1-\pi)W_0+\pi W_1$ and $\Sigma_\pi^{\rm mix}=\diag(f_\pi)-f_\pi f_\pi^\top$.  Then
\begin{equation}\label{eq:rank-one-correction}
        \Sigma_\pi^{\rm mix}=\Sigma_\pi+\pi(1-\pi)vv^\top.
\end{equation}
Consequently $\Sigma_\pi^{\rm mix}\succeq \Sigma_\pi$ on $\mathcal T$.
\end{proposition}

\begin{proof}
Expand $\diag(f_\pi)-f_\pi f_\pi^\top$ using $f_\pi=W_0+\pi v$.  The diagonal terms give $(1-\pi)\diag(W_0)+\pi\diag(W_1)$.  The rank-one terms simplify to $(1-\pi)W_0W_0^\top+\pi W_1W_1^\top+\pi(1-\pi)vv^\top$.  This gives \eqref{eq:rank-one-correction}.
\end{proof}

\begin{proposition}[Relation to the multinomial Fisher proxy]\label[proposition]{prop:fisher-proxy}
Let $\pi\in[0,1]$ and define
\[
        I_f=\sum_{y\in\Y}\frac{v(y)^2}{f_\pi(y)}.
\]
Then
\begin{equation}\label{eq:Ipi-If}
        I_\pi=\frac{I_f}{1-\pi(1-\pi)I_f}.
\end{equation}
In particular, $I_\pi\ge I_f$, with equality if and only if $v=0$ or $\pi\in\{0,1\}$.
\end{proposition}

\begin{proof}
For $u\in\mathcal T$, put $x_y=u(y)/f_\pi(y)$.  Since $\sum_yu(y)=0$,
\[
        \{\diag(f_\pi)-f_\pi f_\pi^\top\}x
        =u-f_\pi\sum_yu(y)=u.
\]
The Moore--Penrose solution differs from $x$ only by a multiple of $\one$, which is orthogonal to $u$.  Therefore
\[
        u^\top(\Sigma_\pi^{\rm mix})^+u
        =u^\top x=\sum_y \frac{u(y)^2}{f_\pi(y)}.
\]
Thus $I_f=v^\top(\Sigma_\pi^{\rm mix})^+v$.  If $\pi\in\{0,1\}$, then \eqref{eq:rank-one-correction} has no rank-one correction and the formula is immediate.  For $\pi\in(0,1)$, \eqref{eq:rank-one-correction} gives, on $\mathcal T$,
\[
        \Sigma_\pi=\Sigma_\pi^{\rm mix}-\pi(1-\pi)vv^\top .
\]
Sherman--Morrison on $\mathcal T$ gives \eqref{eq:Ipi-If}.  Positivity of $\Sigma_\pi$ implies the denominator is positive.  If $v\ne0$ and $\pi\in(0,1)$, then $I_f>0$ and the denominator is strictly smaller than one, so $I_\pi>I_f$.
\end{proof}

\begin{remark}[Canonical case]
For $\pi=0$, $\Sigma_0^+$ acts on $\mathcal T$ as $\diag(1/W_0)$ in quadratic forms.  Hence
\[
        I_0=v^\top\Sigma_0^+v=\sum_y\frac{(W_1(y)-W_0(y))^2}{W_0(y)}=\chi^2(W_1\Vert W_0).
\]
\end{remark}

\begin{proposition}[Endpoint maximum for the tangent Fisher profile]\label[proposition]{prop:boundary-worstcase}
Under the full-support assumption, the function
\[
        \pi\longmapsto I_\pi=v^\top\Sigma_\pi^+v
\]
is twice continuously differentiable and convex on $[0,1]$.  Consequently
\[
        \max_{\pi\in[0,1]} I_\pi
        =\max(I_0,I_1)
        =\max\bigl(\chi^2(W_1\Vert W_0),\,\chi^2(W_0\Vert W_1)\bigr).
\]
Thus the continuum tangent profile has its largest value at one or both endpoints.  The endpoint $\pi=0$ is the canonical $W_0$-baseline orientation, while $\pi=1$ is the corresponding endpoint after interchanging $W_0$ and $W_1$.  This proposition concerns the leading profile $I_\pi$; it does not assert that the exact finite-$n$ maximizer in \eqref{eq:worst-case-dp} is $k=0$ or $k=n-1$.
\end{proposition}

\begin{proof}
Restrict all matrices to the tangent space $\mathcal T=\{x:\one^\top x=0\}$.  Under full support the restrictions
\[
        A_\pi:=\Sigma_\pi|_{\mathcal T}
        =A_0+\pi D,
        \qquad D:=(\Sigma_1-\Sigma_0)|_{\mathcal T},
\]
are positive definite for every $\pi\in[0,1]$.  Since $v\in\mathcal T$,
\[
        I_\pi=v^\top A_\pi^{-1}v.
\]
The inverse differentiation identity
\[
        \frac{\mathrm d}{\mathrm d\pi}A_\pi^{-1}
        =-A_\pi^{-1}DA_\pi^{-1}
\]
gives, after a second differentiation,
\[
\begin{aligned}
        I_\pi''
        &=2v^\top A_\pi^{-1}DA_\pi^{-1}DA_\pi^{-1}v\\
        &=2(DA_\pi^{-1}v)^\top A_\pi^{-1}(DA_\pi^{-1}v)\ge0.
\end{aligned}
\]
Hence $I_\pi$ is convex.  Every convex real-valued function on $[0,1]$ is bounded above by the larger endpoint value, so
$\max_{\pi\in[0,1]}I_\pi=\max(I_0,I_1)$.  The identity for $I_0$ was proved in the preceding remark.  Interchanging $W_0$ and $W_1$ gives
\[
        I_1=\sum_y\frac{(W_0(y)-W_1(y))^2}{W_1(y)}
        =\chi^2(W_0\Vert W_1).
\]
\end{proof}

\begin{lemma}[Uniform inverse and Lipschitz control]\label[lemma]{lem:sigma-lipschitz}
Fix $0<\eta<1/2$ and assume full support.  On every compact interval $J_\eta=[\eta/2,1-\eta/2]$ there are constants $0<c_\eta<C_\eta<\infty$ such that, for all $\pi\in J_\eta$,
\[
        c_\eta I_{\mathcal T}\preceq \Sigma_\pi|_{\mathcal T}\preceq C_\eta I_{\mathcal T},
        \qquad
        \|\Sigma_\pi^+\|_{\mathcal T\to\mathcal T}\le c_\eta^{-1}.
\]
Moreover the pseudoinverse is Lipschitz on $J_\eta$ in the tangent-space operator norm:
\[
        \|\Sigma_\pi^+-\Sigma_{\pi'}^+\|_{\mathcal T\to\mathcal T}
        \le C_\eta |\pi-\pi'|,
        \qquad \pi,\pi'\in J_\eta .
\]
Consequently $\pi\mapsto s_\pi=\Sigma_\pi^+v$ and $\pi\mapsto I_\pi=v^\top\Sigma_\pi^+v$ are uniformly Lipschitz on $J_\eta$.
\end{lemma}

\begin{proof}
On $\mathcal T$ write $A_\pi=\Sigma_\pi|_{\mathcal T}=A_0+\pi D$, where $D=(\Sigma_1-\Sigma_0)|_{\mathcal T}$.  Full support makes $A_\pi$ positive definite for every $\pi\in[0,1]$.  Continuity of the smallest and largest eigenvalues on the compact interval $J_\eta$ gives the displayed uniform spectral bounds.  The inverse differentiation identity
\[
        \frac{d}{d\pi}A_\pi^{-1}=-A_\pi^{-1}DA_\pi^{-1}
\]
therefore gives
\[
        \left\|\frac{d}{d\pi}A_\pi^{-1}\right\|
        \le c_\eta^{-2}\|D\|,
        \qquad \pi\in J_\eta.
\]
Integrating this derivative between $\pi$ and $\pi'$ proves the Lipschitz bound.  The statements for $s_\pi$ and $I_\pi$ follow by multiplying by the fixed vector $v\in\mathcal T$.
\end{proof}

\section{Exact likelihood-ratio identities}\label{sec:lr}

Define
\[
        w(y)=\frac{W_1(y)}{W_0(y)},\qquad r(y)=w(y)-1.
\]

\begin{lemma}[Canonical likelihood ratio]\label[lemma]{lem:canonical-lr}
Assume $\delta_\star>0$.  For every histogram $N$ with $\sum_yN_y=n$,
\[
        L_{n,0}(N):=\frac{\dd T_{n,1}}{\dd T_{n,0}}(N)
        =\frac1n\sum_{y\in\Y}N_y\frac{W_1(y)}{W_0(y)}.
\]
Equivalently, under $T_{n,0}$, if $Y_i\stackrel{\rm i.i.d.}{\sim}W_0$ and $U_n=n^{-1}\sum_{i=1}^nr(Y_i)$, then $L_{n,0}=1+U_n$.
\end{lemma}

\begin{proof}
Under $T_{n,0}$ the histogram mass is
\[
        T_{n,0}(N)=\frac{n!}{\prod_y N_y!}\prod_y W_0(y)^{N_y}.
\]
Under $T_{n,1}$, condition on the changed user's output $y$ and draw the remaining $n-1$ messages from $W_0$.  Dividing the resulting convolution by $T_{n,0}(N)$ gives $\sum_y(N_y/n)W_1(y)/W_0(y)$.
\end{proof}

\begin{lemma}[General neighboring pair]\label[lemma]{lem:general-lr}
Fix $k\in\{0,\ldots,n-1\}$ and assume $\delta_\star>0$.  Let $P=T_{n,k}$ and $Q=T_{n,k+1}$.  Under $P$, let $Y_\star\sim W_0$ be the message of the unique user whose input changes from $0$ to $1$.  Then
\[
        L_{n,k}(N):=\frac{\dd Q}{\dd P}(N)
        =\E\left[\frac{W_1(Y_\star)}{W_0(Y_\star)}\mid N\right]
        =1+\E[r(Y_\star)\mid N].
\]
These identities hold $P$-almost surely.  More explicitly, for every
$N\in\operatorname{supp}P$,
\begin{equation}\label{eq:general-lr-sum}
        L_{n,k}(N)=
        \frac{\sum_y W_1(y)T_{n-1,k}(N-e_y)}{\sum_y W_0(y)T_{n-1,k}(N-e_y)},
\end{equation}
where terms with negative coordinates are interpreted as zero.
Off $\operatorname{supp}P$ we set $L_{n,k}(N)=0$; this is a valid version of
$\dd Q/\dd P$ because $Q\ll P$.
\end{lemma}

\begin{proof}
Condition on $Y_\star=y$.  The remaining $n-1$ users consist of $n-k-1$ zeros and $k$ ones under both hypotheses.  Hence
\[
        P(N)=\sum_y W_0(y)T_{n-1,k}(N-e_y),\qquad
        Q(N)=\sum_y W_1(y)T_{n-1,k}(N-e_y).
\]
If $Q(N)>0$, then for some $y$,
$W_1(y)T_{n-1,k}(N-e_y)>0$.  Since $\delta_\star>0$ gives $W_0(y)>0$,
the corresponding term in $P(N)$ is positive.  Thus $Q\ll P$.  On
$\operatorname{supp}P$, division gives \eqref{eq:general-lr-sum}, which is
Bayes' formula for the conditional expectation.  Outside that support,
positivity of every $W_0(y)$ forces all terms $T_{n-1,k}(N-e_y)$ to vanish,
so the numerator vanishes as well.
\end{proof}

\section{Linearization and residual decomposition}\label{sec:linearization}

Throughout this section assume $\delta_{\rm full}>0$, fixed alphabet size $d$, and an interior composition condition $k/n\in[\eta,1-\eta]$ for some fixed $\eta\in(0,1/2)$.

\subsection{Pointwise conditional-expectation linearization}

The following pointwise statement is useful for intuition and for weak convergence.  The stronger moment decomposition in \cref{subsec:residual-decomp} is used for the JSD expansion and local privacy-curve results.

\begin{theorem}[Conditional-expectation linearization]\label{thm:linearization}
Let $P=T_{n,k}$ and $U_{n,k}=L_{n,k}-1$.  Define $\pi_n=k/(n-1)$ and $s_n=\Sigma_{\pi_n}^+v$.  Let
\[
        A_n=\left\{\|N-\E_PN\|_\infty\le n^{5/8}\right\}.
\]
There are constants $c,C>0$, depending only on $(d,\delta_{\rm full},\eta)$, such that $P(A_n^c)\le C\exp(-cn^{1/4})$ and on $A_n$,
\[
        U_{n,k}=\frac1n s_n^\top(N-\E_PN)+\widetilde R_n(N),
        \qquad \sup_{N\in A_n}|\widetilde R_n(N)|\le Cn^{-3/4}.
\]
\end{theorem}

\begin{proof}
Put $m=n-1$ and write $N=e_{Y_\star}+S$, where $S$ is the histogram of the remaining $m$ users.  Each coordinate of $N-\E_PN$ is a sum of independent centered random variables bounded by one.  Bernstein's inequality and a union bound over the fixed alphabet give
\[
        P(A_n^c)\le 2d\exp\{-c n^{1/4}\}.
\]
Fix the baseline symbol $d$, set $s=N-e_d$, and put $t_y=e_y-e_d$.  On $A_n$, for a fixed constant $A=A(d)$ and all sufficiently large $n$,
\[
        \|(s-\E S)^-\|_2\le A m^{5/8}.
\]
Moreover, every $N-e_y$ has nonnegative coordinates because every coordinate of $\E S$ is at least $m\delta_{\rm full}$.  Hence all lattice probabilities below are positive.  Bayes' formula gives
\[
        \Pp(Y_\star=y\mid N)
        =\frac{W_0(y)\Pp(S=s-t_y)}
        {\sum_z W_0(z)\Pp(S=s-t_z)}.
\]
By \cref{lem:ratio}, uniformly on $A_n$,
\[
        \log\frac{\Pp(S=s-t_y)}{\Pp(S=s)}
        =t_y^\top\theta_n+c_y(N),
        \qquad
        \theta_n=\frac1m\Sigma_{\pi_n}^+(s-\E S),
\]
where $\max_y|c_y(N)|\le Cn^{-3/4}$; the deterministic quadratic term of order $m^{-1}$ is included in $c_y$.  Since $t_y=e_y-e_d$, the common term $-e_d^\top\theta_n$ cancels in the posterior normalization.  Moreover $\|\theta_n\|=O(n^{-3/8})$ on $A_n$.  Taylor expansion of the finite-dimensional softmax, with a remainder bounded uniformly for $\max_y|e_y^\top\theta_n+c_y|$ small, therefore yields
\[
        \Pp(Y_\star=y\mid N)
        =W_0(y)\{1+(e_y-W_0)^\top\theta_n\}+O(n^{-3/4}).
\]
Multiplying by $r(y)$ and summing uses $W_0(y)r(y)=v(y)$ and $\sum_yv(y)=0$, so
\[
        \E[r(Y_\star)\mid N]=v^\top\theta_n+O(n^{-3/4}).
\]
Finally,
\[
        s-\E S=N-\E_PN+(W_0-e_d).
\]
The bounded last term contributes $O(n^{-1})$, and replacing $m^{-1}$ by $n^{-1}$ contributes $O(n^{-11/8})$ on $A_n$.  Thus
\[
        v^\top\theta_n
        =\frac1n s_n^\top(N-\E_PN)+O(n^{-1}),
\]
which proves the assertion.
\end{proof}

\subsection{Regression and residual bounds}\label{subsec:residual-decomp}

The next results reorganize the linearization in $L^p$ rather than pointwise form.  This modular form is the one used later.

\begin{lemma}[Uniform histogram moments]\label[lemma]{lem:hist-moments}
For every integer $p\ge1$ there is $C_p=C_p(d,\delta_{\rm full},\eta)$ such that, uniformly over $k/n\in[\eta,1-\eta]$,
\[
        \E\left\|\frac{N-\E N}{\sqrt n}\right\|_2^p\le C_p.
\]
The same bound holds for histograms of $m$ observations with $m$ in place of $n$ and interior compositions.
\end{lemma}

\begin{proof}
For a fixed symbol $y$, write $N_y-\E N_y=\sum_{i=1}^n V_{i,y}$, where the $V_{i,y}$ are independent, centered, and satisfy $|V_{i,y}|\le1$.  Bernstein's inequality gives, with constants uniform over the allowed compositions,
\[
        \Pp(|N_y-\E N_y|>t)
        \le 2\exp\left\{-c\min\left(\frac{t^2}{n},t\right)\right\},
        \qquad t\ge0.
\]
Using $\E|X|^p=p\int_0^\infty t^{p-1}\Pp(|X|>t)\,dt$ and the change of variables $t=\sqrt n\,u$ gives
\[
        \E|N_y-\E N_y|^p\le C_p n^{p/2}.
\]
For fixed $d$, norm equivalence gives
\[
        \|x\|_2^p\le C_{p,d}\sum_{y=1}^d|x_y|^p,
\]
which proves the claim.  The argument is unchanged with $m$ in place of $n$.
\end{proof}

\begin{lemma}[Exact regression identities]\label[lemma]{lem:regression}
Let $P=T_{n,k}$, $Q=T_{n,k+1}$, $\pi=k/n$, $\Delta=N-\E_PN$, and $U=L_{n,k}-1$.  Define
\[
        s_\pi=\Sigma_\pi^+v,
        \qquad Z=s_\pi^\top\Delta,
        \qquad R=U-\frac{Z}{n}.
\]
Then
\[
        \E_P U=0,
        \qquad \E_P[\Delta U]=v,
        \qquad \Cov_P(N)=n\Sigma_\pi,
\]
and consequently
\[
        \E_P[RZ]=0,
        \qquad \E_P[Z^2]=nI_\pi,
        \qquad \E_P[U^2]=\frac{I_\pi}{n}+\E_P[R^2].
\]
\end{lemma}

\begin{proof}
The identity $\E_PU=0$ follows from $\E_PL=1$.  Since $U=\E[r(Y_\star)\mid N]$ and $\Delta$ is $N$-measurable,
\[
        \E_P[\Delta U]=\E_P[\Delta r(Y_\star)].
\]
Write $N=e_{Y_\star}+S$, where $S$ is independent of $Y_\star$ and $\E r(Y_\star)=0$.  Then only the changed user's contribution remains:
\[
        \E_P[(e_{Y_\star}-W_0)r(Y_\star)]
        =\sum_y W_0(y)r(y)(e_y-W_0)=\sum_y v(y)e_y=v.
\]
The covariance formula is the fixed-composition identity.  Therefore
\[
        \E_P[\Delta R]=\E_P[\Delta U]-\frac1n\E_P[\Delta\Delta^\top]s_\pi
        =v-\Sigma_\pi\Sigma_\pi^+v=0,
\]
which gives $\E_P[RZ]=0$.  The formula for $\E_P[Z^2]$ is $s_\pi^\top(n\Sigma_\pi)s_\pi=nI_\pi$.
\end{proof}

\begin{proposition}[Residual moment bounds]\label[proposition]{prop:residual-bounds}
Under the standing interior and full-support assumptions,
\[
        \sup_{k/n\in[\eta,1-\eta]}\E_P[R^2]=O(n^{-2}),
        \qquad
        \sup_{k/n\in[\eta,1-\eta]}\E_P[R^4]=O(n^{-4}),
\]
and hence $\E_P|R|=O(n^{-1})$, uniformly in $k$.
\end{proposition}

\begin{proof}
We give the proof because this is the key technical estimate.  Let $m=n-1$, $\pi_m=k/m$, and $p=\delta_{\rm full}$.  The remaining-sum array $S$ satisfies the uniform full-support condition (A.0) in \cref{app:edgeworth} with this value of $p$, uniformly over $k/n\in[\eta,1-\eta]$.  Write $N=e_{Y_\star}+S$, where $S$ is the histogram of the remaining $m$ users.  Fix the baseline symbol $d$ and set
\[
        s=N-e_d,
        \qquad p_y(N)=\Pp(S=N-e_y).
\]
For $y\in\Y$, put
\[
        t_y=e_y-e_d.
\]
Then $\one^\top t_y=0$, $\|t_y\|_1\le2$, and
\[
        s-t_y=N-e_d-(e_y-e_d)=N-e_y.
\]
This is the sign convention used below.  With $p_d(N)=\Pp(S=N-e_d)=\Pp(S=s)$, Bayes' formula gives
\[
        \Pp(Y_\star=y\mid N)
        =\frac{W_0(y)p_y(N)}{\sum_z W_0(z)p_z(N)}.
\]
Let
\[
        A_m=\{\|z\|_2\le m^{1/8}\},
        \qquad z=\frac{s^--(\E S)^-}{\sqrt m}.
\]
The refined ratio lemma is used only on $A_m$.  Since $s=S+e_{Y_\star}-e_d$, the complement $A_m^c$ implies
\[
        \|(S-\E S)^-\|_2 \ge m^{5/8}-C_d
\]
for a constant $C_d$.  Bernstein's inequality for the fixed-composition histogram therefore gives
\[
        \Pp(A_m^c)\le C\exp(-c m^{1/4}),
\]
uniformly in $k/n\in[\eta,1-\eta]$.  On $A_m^c$ we use only the crude boundedness $|U|\le C$ and $|Z/n|\le C$, hence $|R|\le C$.  Consequently this tail contributes exponentially little to every fixed moment of $R$.

On $A_m$, every coordinate of $s$ differs from its mean by at most $O(m^{5/8})$.  Since each mean coordinate is at least $mp$, all coordinates of $N-e_y$ are nonnegative for every $y$ once $m$ is large; hence the probabilities in the following ratios are strictly positive.  \Cref{lem:refined-ratio}, applied with $A=1$, gives, for every $y$,
\[
\log\frac{p_y(N)}{p_d(N)}
=\frac1m t_y^\top\Sigma_{\pi_m}^+(s-\E S)
 -\frac{1}{2m}t_y^\top\Sigma_{\pi_m}^+t_y
 +\frac1m Q_{m,t_y}(z)+\rho_{m,t_y}(s),
\]
where $Q_{m,t_y}$ is quadratic with uniformly bounded coefficients and
\[
        |\rho_{m,t_y}(s)|\le C m^{-3/2}(1+\|z\|_2^9).
\]
Set
\[
        \theta=\frac1m\Sigma_{\pi_m}^+(s-\E S)\in\mathcal T,
        \qquad
        a_y=-\frac{1}{2m}t_y^\top\Sigma_{\pi_m}^+t_y+\frac1m Q_{m,t_y}(z)+\rho_{m,t_y}(s).
\]
The bounds above imply
\[
        |a_y|\le C\left(\frac{1+\|z\|_2^2}{m}+\frac{1+\|z\|_2^9}{m^{3/2}}\right).
\]
Since $t_y=e_y-e_d$, the term $-e_d^\top\theta$ is common to all $y$ and cancels in the posterior normalization.  Thus, with $b_y=e_y^\top\theta+a_y$,
\[
        \Pp(Y_\star=y\mid N)
        =\frac{W_0(y)e^{b_y}}{\sum_{\ell} W_0(\ell)e^{b_\ell}}.
\]
On $A_m$, $\max_y|b_y|=O(m^{-3/8})$, so it is smaller than a fixed constant for all sufficiently large $m$.  Taylor's formula, uniformly over the finite alphabet, gives
\[
        e^{b_y}=1+b_y+O(b_y^2),
        \qquad
        \left(\sum_{\ell}W_0(\ell)e^{b_\ell}\right)^{-1}
        =1-\bar b+O(\max_\ell|b_\ell|^2),
\]
where $\bar b=\sum_\ell W_0(\ell)b_\ell$.  Therefore
\[
\begin{aligned}
        \Pp(Y_\star=y\mid N)
        &=W_0(y)\{1+b_y-\bar b\}
          +O(\max_\ell|b_\ell|^2)\\
        &=W_0(y)\{1+(e_y-W_0)^\top\theta\}
          +\varepsilon_y(N).
\end{aligned}
\]
The linear contribution of $a_y-\sum_\ell W_0(\ell)a_\ell$ is placed in $\varepsilon_y$.  Since
$\|\theta\|=O(m^{-1/2}\|z\|_2)+O(m^{-1})$ and the displayed bound on $a_y$ holds,
\[
        |\varepsilon_y(N)|\le
        C\left(\frac{1+\|z\|_2^2}{m}
        +\frac{1+\|z\|_2^9}{m^{3/2}}\right).
\]
The finitely many smaller values of $m$ are absorbed by enlarging $C$.

Multiplying by $r(y)$ and summing over $y$ cancels the constant term.  Since $W_0(y)r(y)=v(y)$ and $\sum_y v(y)=0$,
\[
        U=\E[r(Y_\star)\mid N]
        =v^\top\theta+\mathcal E_n(N),
\]
with
\[
        |\mathcal E_n(N)|\le
        C\left(\frac{1+\|z\|_2^2}{n}+\frac{1+\|z\|_2^9}{n^{3/2}}\right).
\]
Finally,
\[
        s-\E S=N-\E_PN+(W_0-e_d).
\]
Thus
\[
        v^\top\theta
        =\frac1m v^\top\Sigma_{\pi_m}^+(N-\E_PN)+O(n^{-1}).
\]
Replacing $m^{-1}$ by $n^{-1}$ and $\Sigma_{\pi_m}^+$ by $\Sigma_\pi^+$ changes the leading term by
\[
        O\left(\frac{1+\|z\|_2}{n^{3/2}}\right)+O(n^{-1}),
\]
using \cref{lem:sigma-lipschitz} on $[\eta/2,1-\eta/2]$.  Since $|\pi_m-\pi|=O(n^{-1})$, the stated replacement follows.  Therefore, on $A_m$,
\[
        |R|\le C\left(\frac{1+\|z\|_2^2}{n}+\frac{1+\|z\|_2^9}{n^{3/2}}\right).
\]
This estimate is not intended to be a pointwise small-$o$ bound uniformly across the whole moderate window.  The required smallness is obtained after taking moments: the $n^{-1}(1+\|z\|_2^2)$ term has exactly the $L^2$ and $L^4$ sizes needed below, while the $n^{-3/2}(1+\|z\|_2^9)$ term is of still smaller moment order.  Combining this on-window bound with the exponential tail estimate on $A_m^c$ yields
\[
\begin{aligned}
        \E_P R^2
        &\le \frac{C}{n^2}\E(1+\|z\|_2^4)
            +\frac{C}{n^3}\E(1+\|z\|_2^{18})
            +Ce^{-cn^{1/4}}=O(n^{-2}),\\
        \E_P R^4
        &\le \frac{C}{n^4}\E(1+\|z\|_2^8)
            +\frac{C}{n^6}\E(1+\|z\|_2^{36})
            +Ce^{-cn^{1/4}}=O(n^{-4}).
\end{aligned}
\]
Here
\[
        z=\frac{(S-\E S)^-}{\sqrt m}+O(m^{-1/2}),
\]
so \cref{lem:hist-moments} supplies all displayed moments uniformly; in particular, the fourth-moment bound uses fixed standardized moments up to order $36$.  Cauchy--Schwarz gives $\E|R|=O(n^{-1})$.
\end{proof}

\section{Jensen--Shannon and smooth divergence asymptotics}\label{sec:jsd}

The Jensen--Shannon divergence is not itself an $(\varepsilon,\delta)$ guarantee, and we do not use it as one.  It is included because it is the most direct divergence-level window onto the fixed-composition tangent constant $I_\pi$.  The same $I_\pi$ governs the leading term of every smooth $f$-divergence between the neighboring laws (\cref{thm:f-divergence}) and the leading term of the local privacy curve (\cref{thm:local-delta}); identifying it through JSD therefore identifies the constant that controls the whole family.  Two features make JSD the convenient entry point: it is symmetric and bounded, so the score expansion is numerically stable and the leading coefficient is pinned down without one-sided tail effects, and the canonical pair admits an explicit third-order expansion (\cref{thm:canonical-jsd}) that fixes the constant unambiguously before the proportional-composition case.  The privacy-relevant consequences are then read off in \cref{sec:local-privacy}; this section establishes the constant. The divergence also has a direct operational reading.  Let $B$ be uniform on $\{0,1\}$ and let the conditional law of $N$ given $B=0$ and $B=1$ be $T_{n,k}$ and $T_{n,k+1}$, respectively.  The marginal law of $N$ is their midpoint $M$, and therefore
\[
        I(B;N)=\frac12\KL(T_{n,k}\Vert M)
        +\frac12\KL(T_{n,k+1}\Vert M)
        =\JSD(T_{n,k}\Vert T_{n,k+1}).
\]
Thus the leading term $I_\pi/(8n)$ is the conditional mutual-information
leakage rate for a uniform neighbor bit against an observer who knows the
fixed background composition.  This is a specialization of the target-input
mutual information studied by Su--Cheng--Wang
\cite{SuChengWang2025MI}.  For this fixed-composition specialization, the
result here identifies the exact channel- and composition-specific coefficient
and shows that the same $I_\pi$ governs both directed local hockey-stick curves.

\begin{lemma}[Pointwise JSD functional]\label[lemma]{lem:jsd-functional}
Let $P,Q$ be finite distributions with $Q\ll P$ and $L=\dd Q/\dd P$.  Then
\[
        \JSD(P\Vert Q)=\E_P[D(L)],
\]
where
\[
        D(t)=\frac12\log\frac{2}{1+t}+\frac{t}{2}\log\frac{2t}{1+t},\qquad t\ge0,
\]
with $D(0)=\frac12\log2$.
\end{lemma}

\begin{proof}
Use $M=(P+Q)/2=(1+L)P/2$ in the definition of JSD.
\end{proof}

\begin{lemma}[Taylor expansion of $D$]\label[lemma]{lem:D-taylor}
For $|u|\le1/2$,
\[
        D(1+u)=\frac{u^2}{8}-\frac{u^3}{16}+\frac{7u^4}{192}-\frac{3u^5}{128}+O(u^6).
\]
The implicit constant is absolute.
\end{lemma}

\begin{proof}
Substituting $t=1+u$ gives
\[
        D(1+u)
        =\frac12(1+u)\log(1+u)
         -\frac12(2+u)\log(1+u/2).
\]
For $|u|\le1/2$, Taylor's formula for $\log(1+x)$ through degree six has a remainder bounded by $C|x|^7$.  Substitution and collection of coefficients give
\[
        D(1+u)=\frac{u^2}{8}-\frac{u^3}{16}
        +\frac{7u^4}{192}-\frac{3u^5}{128}
        +\frac{31u^6}{1920}+O(u^7).
\]
Dropping the displayed sixth-order term into the remainder proves the stated formula with an absolute constant.
\end{proof}

\begin{theorem}[Canonical JSD expansion]\label{thm:canonical-jsd}
Assume $\delta_\star>0$.  Let $Y\sim W_0$, $r(Y)=W_1(Y)/W_0(Y)-1$, $\chi^2=\E r(Y)^2$, and $\mu_3=\E r(Y)^3$.  Then
\[
        \JSD(T_{n,0}\Vert T_{n,1})
        =\frac{\chi^2}{8n}-\frac{\mu_3}{16n^2}+\frac{7(\chi^2)^2}{64n^2}+O(n^{-3}).
\]
\end{theorem}

\begin{proof}
By \cref{lem:canonical-lr}, $L=1+U_n$ with $U_n=n^{-1}\sum_{i=1}^nr(Y_i)$ under $T_{n,0}$.  Let $X=r(Y)$ under $Y\sim W_0$.  Then $\E X=0$, $\E X^2=\chi^2$, and $X$ is bounded under $\delta_\star>0$.  Writing $\mu_j=\E X^j$, independence and centering give the exact identities
\[
\begin{aligned}
        \E U_n^2&=\frac{\chi^2}{n},\\
        \E U_n^3&=\frac{\mu_3}{n^2},\\
        \E U_n^4&=\frac{n\mu_4+3n(n-1)(\chi^2)^2}{n^4}
        =\frac{3(\chi^2)^2}{n^2}+O(n^{-3}),\\
        \E U_n^5&=\frac{n\mu_5+10n(n-1)\mu_3\chi^2}{n^5}
        =O(n^{-3}).
\end{aligned}
\]
Rosenthal's inequality at order six gives $\E|\sum_iX_i|^6=O(n^3)$, hence $\E|U_n|^6=O(n^{-3})$.

Let $B_n=\{|U_n|\le1/2\}$.  Hoeffding's inequality gives $\Pp(B_n^c)\le Ce^{-cn}$.  On $B_n$ apply \cref{lem:D-taylor}; outside $B_n$ the contribution is exponentially small because $D(1+U_n)$ is bounded on the compact range of $1+U_n$.  Hence
\[
\begin{aligned}
\JSD(T_{n,0}\Vert T_{n,1})
&=\frac18\E U_n^2-\frac1{16}\E U_n^3+\frac7{192}\E U_n^4
  -\frac3{128}\E U_n^5+O(\E|U_n|^6)+O(e^{-cn})\\
&=\frac{\chi^2}{8n}-\frac{\mu_3}{16n^2}+\frac{7(\chi^2)^2}{64n^2}+O(n^{-3}).
\end{aligned}
\]
\end{proof}

\begin{theorem}[Interior fixed-composition JSD expansion]\label{thm:interior-jsd}
Assume $\delta_{\rm full}>0$ and fix $\eta\in(0,1/2)$.  Uniformly for all $k$ with $k/n\in[\eta,1-\eta]$,
\[
        \JSD(T_{n,k}\Vert T_{n,k+1})
        =\frac{I_{k/n}}{8n}+O(n^{-2}).
\]
\end{theorem}

\begin{proof}
Let $P=T_{n,k}$, $Q=T_{n,k+1}$, $U=L-1$, and use the decomposition $U=Z/n+R$ from \cref{lem:regression,prop:residual-bounds}.  The exact regression identity gives
\[
        \E_P U^2=\frac{I_{k/n}}{n}+O(n^{-2}).
\]
We next record the moments needed for Taylor expansion.  Write $A=Z/n$.  The score $Z=\sum_{i=1}^nX_i$ is a sum of independent centered random variables whose absolute values are uniformly bounded.  Hence Rosenthal's inequality gives, for each fixed $p\ge2$,
\[
        \E_P|A|^p=O(n^{-p/2}).
\]
For the signed third moment no absolute-moment inequality is used: independence and centering give
\[
        \E_PZ^3=\sum_{i=1}^n\E_PX_i^3=O(n),
\]
so $\E_PA^3=O(n^{-2})$.  Direct expansion also gives $\E_PA^4=O(n^{-2})$.

The variables $U$, $A$, and $R=U-A$ are uniformly bounded under full support and interior composition.  Using $\E R^2=O(n^{-2})$ and $\E R^4=O(n^{-4})$, the mixed terms satisfy
\[
\begin{aligned}
 |\E A^2R|&\le(\E A^4)^{1/2}(\E R^2)^{1/2}=O(n^{-2}),\\
 |\E AR^2|&\le(\E A^2)^{1/2}(\E R^4)^{1/2}=O(n^{-5/2}),\\
 \E|R|^3&\le(\E R^2\,\E R^4)^{1/2}=O(n^{-3}).
\end{aligned}
\]
Therefore $\E U^3=O(n^{-2})$.  Similarly, $\E U^4=O(n^{-2})$.  Finally,
\[
        \E|U|^5\le16\{\E|A|^5+\E|R|^5\}=O(n^{-5/2}),
\]
because $R$ is uniformly bounded and $\E|R|^5\le C\E R^4=O(n^{-4})$.  The same boundedness, together with $\E U^4=O(n^{-2})$, also gives $\E|U|^6=O(n^{-2})$.

To justify the Taylor event, write $U=Z/n+R$.  Bernstein's inequality gives $\Pp(|Z|>n/4)\le Ce^{-cn}$, while \cref{prop:residual-bounds} gives $\Pp(|R|>1/4)\le C\E R^4=O(n^{-4})$.  Hence
\[
        \Pp(|U|>1/2)=O(n^{-4}).
\]
On $\{|U|\le1/2\}$ apply \cref{lem:D-taylor}; on the complement use boundedness of $D(1+U)$ under full support.  The tail contribution is $O(n^{-4})$, and therefore
\[
        \E_P[D(1+U)]=\frac18\E_PU^2+O(n^{-2})
        =\frac{I_{k/n}}{8n}+O(n^{-2}).
\]
\end{proof}

\begin{theorem}[Smooth $f$-divergence constant]\label{thm:f-divergence}
Let $J$ be a compact interval containing the full-support likelihood-ratio range, uniformly over the interior compositions under consideration.  Let $f$ be convex with $f(1)=0$, continuous on an open interval containing $J$, and three times continuously differentiable in a neighborhood of $1$.  We do not require $f'(1)=0$.  Under the assumptions of \cref{thm:interior-jsd},
\[
        D_f(T_{n,k+1}\Vert T_{n,k})
        =\frac{f''(1)}{2}\frac{I_{k/n}}{n}+o(n^{-1}),
\]
uniformly over $k/n\in[\eta,1-\eta]$ when $f$ and its derivatives are fixed.  In particular, for fixed R\'enyi order $\alpha>1$,
\[
        D_\alpha(T_{n,k+1}\Vert T_{n,k})=\frac{\alpha I_{k/n}}{2n}+o(n^{-1}).
\]
\end{theorem}

\begin{proof}
Let $L=1+U=dT_{n,k+1}/dT_{n,k}$.  A linear term in $f$ does not change the divergence because
\[
        \E_{T_{n,k}}[L-1]=0.
\]
Thus, with
\[
        \widetilde f(x):=f(x)-f'(1)(x-1),
\]
we have $D_f=D_{\widetilde f}$, $\widetilde f(1)=\widetilde f'(1)=0$, and $\widetilde f''(1)=f''(1)$.
Choose $a>0$ so that $\widetilde f$ is $C^3$ on $[1-a,1+a]$.  On the event $G_a=\{|U|\le a\}$,
\[
        \widetilde f(1+U)=\frac12 f''(1)U^2+O(|U|^3).
\]
On $G_a^c$, use $U=Z/n+R$: Bernstein's inequality gives $\Pp(|Z|>an/2)\le Ce^{-cn}$ and \cref{prop:residual-bounds} gives $\Pp(|R|>a/2)=O(n^{-4})$.  Since $L$ is confined to the fixed compact interval $J$ under full support and $f$ is continuous, hence bounded, on $J$, the contribution of $G_a^c$ is $O(n^{-4})$.  Hence
\[
        D_f=\frac12 f''(1)\E U^2+O(\E|U|^3)+O(n^{-4}),
\]
and the claim follows from $\E U^2=I_{k/n}/n+o(n^{-1})$ and $\E|U|^3=o(n^{-1})$.  For R\'enyi divergence, apply this to
\[
        f_\alpha(x)=\frac{x^\alpha-1}{\alpha-1},
\]
for which $f_\alpha''(1)=\alpha$, and then use
\[
        D_\alpha=(\alpha-1)^{-1}\log\{1+(\alpha-1)D_{f_\alpha}\}
        =D_{f_\alpha}+O(D_{f_\alpha}^2).
\]
\end{proof}

\section{Gaussian tangent privacy curves}\label{sec:local-privacy}

This section replaces a single global GDP claim by two separate statements.  First, the privacy loss has a Gaussian tangent expansion after normalization.  Second, the one-sided DP curve has a local-scale asymptotic with an $O(n^{-1})$ error when $\eps=t\mu_n$.  The latter is the privacy-curve statement with an error below the leading privacy signal.

\subsection{Normalized privacy loss}

Let $P=T_{n,k}$, $Q=T_{n,k+1}$, $\pi=k/n$, and define
\[
        \mu_n=\sqrt{I_\pi/n},
        \qquad
        \xi_n=\frac{Z}{\sqrt{nI_\pi}},
        \qquad Z=s_\pi^\top(N-\E_PN).
\]
Then $Z/n=\mu_n\xi_n$.

The following normalized Gaussian tangent statement is included for
interpretation.  The local hockey-stick theorem below does not use it as a
black-box trade-off theorem; it uses directly the stronger residual
decomposition $L-1=Z/n+R$ together with the call-payoff approximation for the
score.

\begin{theorem}[Interior Gaussian tangent expansion with normalized remainder]\label{thm:LAN-normalized}
Assume $\delta_{\rm full}>0$, $I_\pi>0$, and $k/n\in[\eta,1-\eta]$.  Under $P=T_{n,k}$,
\begin{equation}\label{eq:LAN-strong}
        \Lambda_{n,k}
        =\mu_n\xi_n-\frac12\mu_n^2+\rho_n,
        \qquad
        \E_P|\rho_n|=O(n^{-1}),
        \qquad
        \E_P\rho_n=O(n^{-2}).
\end{equation}
In particular,
\[
        \frac{\Lambda_{n,k}}{\mu_n}
        =\xi_n-\frac12\mu_n+o_{L^1(P)}(1),
\]
uniformly over the stated interior compositions.
Moreover $\xi_n$ satisfies a Berry--Esseen bound
\[
        \sup_x\left|P(\xi_n\le x)-\PhiG(x)\right|\le Cn^{-1/2},
\]
with $C$ depending only on $(d,\delta_{\rm full},\eta)$ and the channel.
\end{theorem}

\begin{proof}
Write $U=L-1=Z/n+R=\mu_n\xi_n+R$.  By \cref{prop:residual-bounds}, $\E|R|=O(n^{-1})$ and $\E R^2=O(n^{-2})$.  Under full support, $L=1+U=\E[W_1(Y_\star)/W_0(Y_\star)\mid N]$ lies in a fixed compact subinterval of $(0,\infty)$, uniformly in $n$ and $k$.  Therefore the Taylor remainder for $\log(1+u)$ is globally controlled on the range of $U$:
\[
        \left|\log(1+U)-U+\frac12U^2\right|\le C|U|^3.
\]
Put $A=Z/n=\mu_n\xi_n$.  Rosenthal's inequality gives $\E|A|^3=O(n^{-3/2})$, while
\[
        \E|R|^3\le(\E R^2\,\E R^4)^{1/2}=O(n^{-3}).
\]
Hence $\E|U|^3=O(n^{-3/2})$ and
\[
        \log(1+U)=U-\frac12U^2+O_{L^1}(n^{-3/2}).
\]
Subtracting $A-\mu_n^2/2$ from the quadratic approximation gives the exact decomposition
\[
        R-\frac12(A^2-\mu_n^2)-AR-\frac12R^2.
\]
Its $L^1$ norm is $O(n^{-1})$: the first term is controlled by $\E|R|=O(n^{-1})$; for the second,
\[
        \E|A^2-\mu_n^2|
        \le \Var(A^2)^{1/2}=O(n^{-1}),
\]
because $\E Z^4=O(n^2)$; and
$\E|AR|\le(\E A^2\,\E R^2)^{1/2}=O(n^{-3/2})$, while $\E R^2=O(n^{-2})$.  This proves the $L^1$ part of \eqref{eq:LAN-strong}.  To identify the mean of the remainder, use the same Taylor-event argument as in \cref{thm:interior-jsd}.  The moment estimates there give
\[
        \E U=0,\qquad
        \E U^2=\mu_n^2+O(n^{-2}),\qquad
        \E U^3=O(n^{-2}),\qquad
        \E U^4=O(n^{-2}),
\]
while $\E|U|^5=O(n^{-5/2})$ and the complement of $\{|U|\le1/2\}$ contributes $O(n^{-4})$.  Expanding $\log(1+U)$ through fourth order therefore yields
\[
        \E_P\Lambda_{n,k}=-\frac12\mu_n^2+O(n^{-2}).
\]
Since $\E\xi_n=0$, this is exactly $\E_P\rho_n=O(n^{-2})$.

Finally write $Z=\sum_iX_{n,i}$ with independent centered uniformly bounded summands.  The total variance is $nI_{k/n}$ and
\[
        \frac{\sum_i\E|X_{n,i}|^3}{(nI_{k/n})^{3/2}}=O(n^{-1/2})
\]
uniformly on the compact composition interval.  The non-uniform Berry--Esseen inequality \cite[Theorem~2.2]{ChenShao2001}, used explicitly in \cref{lem:call-payoff}, implies the stated Kolmogorov bound after taking the supremum over $x$.
\end{proof}

\subsection{Local-scale privacy curve}

For a Gaussian shift experiment $(N(0,1),N(\mu,1))$, the one-sided GDP curve is
\begin{equation}\label{eq:gdp-delta}
        \delta_{\rm GDP}(\eps;\mu)
        =\PhiG\left(-\frac{\eps}{\mu}+\frac{\mu}{2}\right)
        -e^\eps\PhiG\left(-\frac{\eps}{\mu}-\frac{\mu}{2}\right).
\end{equation}
On the local scale $\eps=t\mu$,
\begin{equation}\label{eq:gdp-local-expansion}
        \delta_{\rm GDP}(t\mu;\mu)
        =\mu\{\phiG(t)-t\PhiG(-t)\}+O(\mu^2),
\end{equation}
uniformly for $t$ in compact subsets of $[0,\infty)$.  Indeed, for $0\le t\le T$,
\[
\begin{aligned}
        \PhiG(-t+\mu/2)&=\PhiG(-t)+\frac\mu2\phiG(t)+O_T(\mu^2),\\
        \PhiG(-t-\mu/2)&=\PhiG(-t)-\frac\mu2\phiG(t)+O_T(\mu^2),\\
        e^{t\mu}&=1+t\mu+O_T(\mu^2),
\end{aligned}
\]
and substitution in \eqref{eq:gdp-delta} gives \eqref{eq:gdp-local-expansion}.

\begin{lemma}[Call-payoff approximation for the normalized score]\label[lemma]{lem:call-payoff}
Under the assumptions of \cref{thm:LAN-normalized}, for each fixed $T<\infty$,
\[
        \sup_{0\le t\le T}\left|
        \E_P[(\xi_n-t)_+]-\E[(G-t)_+]
        \right|\le C_T n^{-1/2},
        \qquad G\sim N(0,1).
\]
\end{lemma}

\begin{proof}
Write $\xi_n=\sum_i\widetilde X_{n,i}$ with
\[
        \widetilde X_{n,i}=X_{n,i}/B_n,
        \qquad B_n^2=\sum_i\Var(X_{n,i})=nI_{k/n}.
\]
The variables $X_{n,i}$ are independent, centered, and uniformly bounded.  Because $\pi\mapsto I_\pi$ is continuous and positive for the fixed nontrivial channel, its infimum on $[\eta,1-\eta]$ is positive.  Hence
\[
        \max_i|\widetilde X_{n,i}|=O(n^{-1/2}),
        \qquad
        \sum_i\E|\widetilde X_{n,i}|^3=O(n^{-1/2}),
\]
uniformly over the stated compositions.

Chen--Shao \cite[Theorem~2.2]{ChenShao2001} applies to the independent normalized summands, whose total variance is one.  For all sufficiently large $n$, $|\widetilde X_{n,i}|\le1$, so the truncated second-moment term in their theorem is zero for every $x$, while the remaining term is bounded by
\[
        \frac{C}{(1+|x|)^3}
        \sum_i\E|\widetilde X_{n,i}|^3.
\]
After increasing $C$ to absorb the finitely many smaller values of $n$, we obtain
\[
        \left|\Pp(\xi_n\le x)-\PhiG(x)\right|
        \le \frac{C n^{-1/2}}{1+|x|^3},\qquad x\in\R,
\]
uniformly over $k/n\in[\eta,1-\eta]$.  For $t\ge0$,
\[
        \E[(\xi_n-t)_+]=\int_t^\infty \Pp(\xi_n>u)\,du,
        \qquad
        \E[(G-t)_+]=\int_t^\infty \Pp(G>u)\,du.
\]
Integrating the non-uniform bound gives
\[
        \left|\E[(\xi_n-t)_+]-\E[(G-t)_+]\right|
        \le Cn^{-1/2}\int_t^\infty \frac{du}{1+u^3}
        \le Cn^{-1/2},
\]
uniformly for $0\le t\le T$.
\end{proof}

\begin{theorem}[Local one-sided privacy curve]\label{thm:local-delta}
Assume $\delta_{\rm full}>0$, $I_{k/n}>0$, and $k/n\in[\eta,1-\eta]$.  Let $\mu_n=\sqrt{I_{k/n}/n}$.  For each fixed $T<\infty$,
\begin{equation}\label{eq:local-delta-main}
        \sup_{0\le t\le T}\left|
        \delta_{T_{n,k+1}\Vert T_{n,k}}(t\mu_n)
        -\mu_n\{\phiG(t)-t\PhiG(-t)\}
        \right|
        \le C_T n^{-1},
\end{equation}
where $C_T$ depends on $(T,d,\delta_{\rm full},\eta)$ and the channel, but not on $n$ or $k$.
\end{theorem}

\begin{proof}
By the positive-part formula,
\[
        \delta_{Q\Vert P}(t\mu_n)=\E_P[(1+U-e^{t\mu_n})_+].
\]
Set $a_n(t)=e^{t\mu_n}-1=t\mu_n+O_T(\mu_n^2)$.  Since $x\mapsto (x-a)_+$ is $1$-Lipschitz,
\[
\left|\E[(U-a_n(t))_+]-\E[(\mu_n\xi_n-a_n(t))_+]\right|
\le \E|R|=O(n^{-1}).
\]
Replacing $a_n(t)$ by $t\mu_n$ changes the expectation by at most $O_T(\mu_n^2)=O_T(n^{-1})$.  Thus
\[
        \delta_{Q\Vert P}(t\mu_n)
        =\mu_n\E[(\xi_n-t)_+]+O_T(n^{-1}).
\]
By \cref{lem:call-payoff}, uniformly for $t\in[0,T]$,
\[
        \E[(\xi_n-t)_+]=\E[(G-t)_+]+O(n^{-1/2}),
        \qquad G\sim N(0,1).
\]
Multiplying by $\mu_n=O(n^{-1/2})$ gives an $O(n^{-1})$ contribution.  Finally,
\[
        \E[(G-t)_+]=\int_t^\infty (x-t)\phiG(x)\,dx=\phiG(t)-t\PhiG(-t).
\]
\end{proof}

\begin{corollary}[Local GDP agreement]\label[corollary]{cor:local-gdp-agreement}
Under the assumptions of \cref{thm:local-delta}, uniformly for $0\le t\le T$,
\[
        \delta_{T_{n,k+1}\Vert T_{n,k}}(t\mu_n)
        =\delta_{\rm GDP}(t\mu_n;\mu_n)+O_T(n^{-1}).
\]
\end{corollary}

\begin{proof}
Combine \cref{thm:local-delta} with \eqref{eq:gdp-local-expansion}.
\end{proof}

\begin{corollary}[Two-sided local privacy curve]\label[corollary]{cor:two-sided}
Under the assumptions of \cref{thm:local-delta}, the reverse ordering obeys the same expansion: for each fixed $T<\infty$,
\[
        \sup_{0\le t\le T}\left|
        \delta_{T_{n,k}\Vert T_{n,k+1}}(t\mu_n)
        -\mu_n\{\phiG(t)-t\PhiG(-t)\}
        \right|\le C_T n^{-1}.
\]
Hence both orderings of the neighboring pair are governed by the same Gaussian tangent curve, and at $\eps=t\mu_n$ the pair is $(\eps,\delta)$-indistinguishable with $\delta=\mu_n\{\phiG(t)-t\PhiG(-t)\}+O_T(n^{-1})$.
\end{corollary}

\begin{proof}
Let $P=T_{n,k}$, $Q=T_{n,k+1}$, $L=\dd Q/\dd P=1+U$.  Since $L>0$, the elementary identity $L\,(L^{-1}-a)_+=(1-aL)_+$ and the change of measure $\E_Q[g]=\E_P[Lg]$ give, with $a=e^{t\mu_n}$,
\[
        \delta_{P\Vert Q}(t\mu_n)
        =\E_Q[(L^{-1}-a)_+]
        =\E_P[L\,(L^{-1}-a)_+]
        =\E_P[(1-e^{t\mu_n}(1+U))_+].
\]
Put $a=e^{t\mu_n}$.  First compare $U=\mu_n\xi_n+R$ with its score part.  Since $x\mapsto(1-a(1+x))_+$ is $a$-Lipschitz,
\[
\left|
\E_P[(1-a(1+U))_+]
-\E_P[(1-a(1+\mu_n\xi_n))_+]
\right|
\le a\E_P|R|=O_T(n^{-1}).
\]
For the remaining comparison, the two affine arguments differ by
\[
\begin{aligned}
&\{1-a-a\mu_n\xi_n\}-\{-t\mu_n-\mu_n\xi_n\}\\
&\qquad=1-a+t\mu_n+(1-a)\mu_n\xi_n.
\end{aligned}
\]
The positive-part map is $1$-Lipschitz, $|1-a+t\mu_n|=O_T(\mu_n^2)$, $|1-a|=O_T(\mu_n)$, and $\E_P|\xi_n|\le(\E_P\xi_n^2)^{1/2}=1$.  Therefore
\[
\left|
\E_P[(1-a(1+\mu_n\xi_n))_+]
-\E_P[(-t\mu_n-\mu_n\xi_n)_+]
\right|=O_T(\mu_n^2)=O_T(n^{-1}).
\]
Consequently
\[
        \delta_{P\Vert Q}(t\mu_n)
        =\mu_n\E_P[(-\xi_n-t)_+]+O_T(n^{-1}).
\]
Apply the Chen--Shao bound used in \cref{lem:call-payoff} directly to the independent summands $-X_{n,i}$.  Their variances and third absolute moments are unchanged, so, uniformly for $t\in[0,T]$,
\[
        \E_P[(-\xi_n-t)_+]=\E[(G-t)_+]+O(n^{-1/2}),\qquad G\sim N(0,1).
\]
Multiplying by $\mu_n=O(n^{-1/2})$ and using $\E[(G-t)_+]=\phiG(t)-t\PhiG(-t)$ proves the expansion; the indistinguishability statement combines it with \cref{thm:local-delta}.
\end{proof}

\subsection{Two-sided curves and the limits of global Gaussian approximation}

\Cref{thm:local-delta,cor:two-sided} together give the two-sided privacy curve.  \Cref{thm:local-delta} controls $\delta_{Q\Vert P}$ directly under $P=T_{n,k}$, and \cref{cor:two-sided} controls the reverse ordering $\delta_{P\Vert Q}$ through the exact change of measure $\E_Q[g]=\E_P[Lg]$ together with the identity $L\,(L^{-1}-a)_+=(1-aL)_+$, which expresses the reverse curve as an expectation under $P$ of a function of the same score whose Lipschitz constant is $e^\eps$ (and hence is uniformly bounded on the local scale $\eps=t\mu_n$ for $t$ in compact sets).

No quantitative transfer of Berry--Esseen rates to the alternative law is needed.  A direct normal approximation under $Q$ would require a correlation bound on $\E_P[R h]$ for the residual $R$ in $L-1=Z/n+R$; the change of measure avoids it, because it removes $R$ from the integrand at the level of the Lipschitz comparison rather than inside a transferred limit theorem.  The reduction is therefore algebraic, not analytic, and the resulting two-sided statement holds uniformly for interior compositions $k/n\in[\eta,1-\eta]$ and $t$ in compact sets, using only \cref{prop:residual-bounds} and \cref{lem:call-payoff}.

\begin{remark}[Certification warning]
A global GDP approximation with an absolute error of order $O(n^{-1/2})$, even when available, has the same order as $\mu_n$.  Therefore it does not identify the leading deviation of the trade-off curve from the trivial curve $1-\alpha$ at fixed $\alpha$, nor does it certify the fixed-$\eps$ privacy curve when the relevant $\delta$ is much smaller than $n^{-1/2}$.  Finite-$n$ certification must use \cref{sec:finite-accounting} or another explicit upper bound.
\end{remark}

\section{\texorpdfstring{Exact finite-$n$ accounting identities}{Exact finite-n accounting identities}}\label{sec:finite-accounting}

The positive-part identity follows directly from the Neyman--Pearson maximizing event.  The contribution here is its specialization to shuffled histogram experiments and its role as the finite-$n$ layer of the analysis.  The formulas in this section are exact finite sums; any numerical use therefore reduces to computing or bounding the histogram probabilities and the displayed likelihood ratios.

\begin{proposition}[Positive-part formulas]\label[proposition]{prop:positive-part}
For finite distributions $P,Q$ with $Q\ll P$ and $L=\dd Q/\dd P$,
\[
        \delta_{Q\Vert P}(\eps)=\E_P[(L-e^\eps)_+].
\]
Moreover,
\begin{equation}\label{eq:reverse-positive-part}
\begin{aligned}
        \delta_{P\Vert Q}(\eps)
        &=\E_P[(1-e^\eps L)_+]\\
        &=P\{L=0\}
          +\E_Q\!\left[(L^{-1}-e^\eps)_+{\bf 1}_{\{L>0\}}\right].
\end{aligned}
\end{equation}
The maximizing events may be taken to be $\{L>e^\eps\}$ and
$\{L<e^{-\eps}\}$, respectively, with arbitrary randomization on ties.
\end{proposition}

\begin{proof}
For any event $A$,
\[
        Q(A)-e^\eps P(A)=\E_P[(L-e^\eps){\bf 1}_A]
        \le \E_P[(L-e^\eps)_+].
\]
Equality is attained by $A=\{L>e^\eps\}$.
Similarly,
\[
        P(A)-e^\eps Q(A)
        =\E_P[(1-e^\eps L){\bf 1}_A]
        \le \E_P[(1-e^\eps L)_+].
\]
Splitting the last expectation over $\{L=0\}$ and $\{L>0\}$ and using
$\dd Q=L\,\dd P$ on $\{L>0\}$ gives the second line of
\eqref{eq:reverse-positive-part}.
\end{proof}

\begin{corollary}[General fixed-composition accounting formulas]\label[corollary]{cor:general-accounting}
Assume $\delta_\star>0$, and put $P=T_{n,k}$ and $Q=T_{n,k+1}$.  Then
\[
        \delta_{Q\Vert P}(\eps)
        =\sum_{\substack{N:\,\sum_yN_y=n\\T_{n,k}(N)>0}}
        T_{n,k}(N)\left(L_{n,k}(N)-e^\eps\right)_+,
\]
where $L_{n,k}(N)$ is given by \eqref{eq:general-lr-sum}.  The reverse
direction is evaluated on the same support without requiring $P\ll Q$:
\[
        \delta_{P\Vert Q}(\eps)
        =\sum_{\substack{N:\,\sum_yN_y=n\\T_{n,k}(N)>0}}
        T_{n,k}(N)\left(1-e^\eps L_{n,k}(N)\right)_+.
\]
The two-sided and worst-case shuffle curves are obtained by taking the maxima
in \eqref{eq:worst-case-dp}.
\end{corollary}

\begin{algorithm}[H]
\caption{Exact finite-$n$ accounting for a fixed pair $(T_{n,k},T_{n,k+1})$}
\label{alg:accounting}
\begin{enumerate}[leftmargin=2.2em,label=\arabic*.,itemsep=0.25em]
\item \textbf{Input:} channel $(W_0,W_1)$, sample size $n$, composition
$k$, and privacy level $\eps$.
\item Compute the histogram probability tables for $T_{n,k}$ and
$T_{n-1,k}$ by dynamic programming over the alphabet.
\item Initialize $\Delta_{Q\Vert P}\gets0$ and
$\Delta_{P\Vert Q}\gets0$.
\item For each histogram $N$ with $|N|=n$ and $T_{n,k}(N)>0$:
  \begin{enumerate}[leftmargin=2em,label=(\alph*),itemsep=0.15em]
  \item compute $L_{n,k}(N)$ from \eqref{eq:general-lr-sum};
  \item add $T_{n,k}(N)(L_{n,k}(N)-e^\eps)_+$ to
  $\Delta_{Q\Vert P}$;
  \item add $T_{n,k}(N)(1-e^\eps L_{n,k}(N))_+$ to
  $\Delta_{P\Vert Q}$.
  \end{enumerate}
\item Return $(\Delta_{Q\Vert P},\Delta_{P\Vert Q})$ as the two directed
curves.
\end{enumerate}
\end{algorithm}

For fixed alphabet size, the dynamic program has polynomially many histogram states in $n$.  This is separate from the asymptotic expansion above: it is an exact accounting representation, not an appeal to the Gaussian tangent approximation.

\subsection{Binary-output specialization}

Assume $\Y=\{0,1\}$ with $p_0=W_0(1)\in(0,1)$ and
$p_1=W_1(1)\in[0,1]$.  Under $T_{n,0}$,
$K=N_1\sim\operatorname{Bin}(n,p_0)$ and
\[
        L_{n,0}(K)=\frac{n-K}{n}\frac{1-p_1}{1-p_0}+\frac{K}{n}\frac{p_1}{p_0}.
\]
Thus
\begin{equation}\label{eq:binary-delta}
        \delta_{T_{n,1}\Vert T_{n,0}}(\eps)
        =\sum_{j=0}^n \binom{n}{j}p_0^j(1-p_0)^{n-j}
        \left(L_{n,0}(j)-e^\eps\right)_+.
\end{equation}
For the reverse direction of the same canonical pair, one must use the same likelihood ratio rather than interchange $p_0$ and $p_1$.  Namely,
\begin{equation}\label{eq:binary-delta-reverse}
        \delta_{T_{n,0}\Vert T_{n,1}}(\eps)
        =\sum_{j=0}^n \binom{n}{j}p_0^j(1-p_0)^{n-j}
        \left(1-e^\eps L_{n,0}(j)\right)_+.
\end{equation}
When $T_{n,0}\ll T_{n,1}$ as well (in particular, under full support), the
reverse curve can equivalently be written as
\[
        \delta_{T_{n,0}\Vert T_{n,1}}(\eps)
        =\sum_{j=0}^n T_{n,1}(j)
        \left(L_{n,0}(j)^{-1}-e^\eps\right)_+.
\]
Without mutual absolute continuity the exact change-of-measure formula is
instead
\[
\begin{aligned}
        \delta_{T_{n,0}\Vert T_{n,1}}(\eps)
        &=T_{n,0}\{L_{n,0}=0\}\\
        &\quad+\sum_{j:\,T_{n,1}(j)>0}T_{n,1}(j)
        \left(L_{n,0}(j)^{-1}-e^\eps\right)_+.
\end{aligned}
\]
A formal interchange of $p_0$ and $p_1$ describes the opposite canonical experiment with baseline $W_1$, not the reverse privacy curve for the fixed pair $(T_{n,0},T_{n,1})$.

\subsection{Chernoff upper bound for the canonical pair}

The bound below is an elementary exponential relaxation of the positive-part
identity of \cref{prop:positive-part}, stated here only because it is
closed-form and requires no numerical transform.  It is not competitive with,
and is not offered as an alternative to, the decomposition-based computable
bounds of Su--Cheng--Wang \cite{SuChengWang2025Decomposition} or the
characteristic-function accountants of
\cite{KoskelaJalkoHonkela2020,ZhuWang2021,KoskelaHeikkilaHonkela2023}, all of
which are tighter at the cost of numerical evaluation.  No novelty is claimed
for it.

Let $Z=r(Y)$ for $Y\sim W_0$ and $M(\lambda)=\E e^{\lambda Z}$.  Let $M'(\lambda)=\E[Ze^{\lambda Z}]$ and $\tau=e^\eps-1$.

\begin{proposition}[Chernoff-type upper bound]\label[proposition]{prop:chernoff}
For the canonical pair and every $\lambda>0$,
\[
        \delta_{T_{n,1}\Vert T_{n,0}}(\eps)
        \le e^{-\lambda n\tau}M(\lambda)^{n-1}\{M(\lambda)+M'(\lambda)\}.
\]
\end{proposition}

\begin{proof}
By \cref{prop:positive-part,lem:canonical-lr},
\[
        \delta=\E[(U_n-\tau)_+]\le \E[(1+U_n){\bf 1}\{U_n\ge\tau\}].
\]
On $\{U_n\ge\tau\}$, $e^{\lambda n(U_n-\tau)}\ge1$.  With $S_n=nU_n=\sum_iZ_i$,
\[
        \E[(1+U_n)e^{\lambda n(U_n-\tau)}]
        =e^{-\lambda n\tau}\left(M(\lambda)^n+M'(\lambda)M(\lambda)^{n-1}\right).
\]
\end{proof}

\subsection{Numerical illustration and comparison}\label{subsec:numerics}

We illustrate the constant $I_\pi$ on a three-symbol channel, where the fixed-composition correction is genuinely present, and record the binary-output case separately as the variance-symmetric example.  We then compare the local privacy curve with two reference objects: the i.i.d.\ mixture heuristic and published mechanism-wide amplification bounds.

All numerical entries below are evaluations of formulas stated either in this manuscript or in the cited source.  To specify the finite-sum computation completely, order the $n$ categorical laws as $p_1,\ldots,p_n$, set $H_0(0)=1$ and $H_0(a)=0$ for $a\ne0$, and iterate
\[
        H_j(a)=\sum_{y\in\Y}p_j(y)H_{j-1}(a-e_y),
        \qquad j=1,\ldots,n.
\]
Here $H_{j-1}(b)=0$ whenever $b$ has a negative coordinate.  For $T_{n,k}$, the list contains $n-k$ copies of $W_0$ and $k$ copies of $W_1$, and $H_n$ is the exact histogram table.  Jensen--Shannon and hockey-stick divergences are then evaluated by their defining finite sums.  The reported boundary-pair $\eps$ is obtained by bisection on the maximum of the two exact directed curves, with final bracket width below $10^{-12}$.  For $I_\pi$, if $B$ has orthonormal columns spanning $\mathcal T$, we evaluate $(B^\top v)^\top(B^\top\Sigma_\pi B)^{-1}(B^\top v)$.  Values are rounded only after these evaluations.

\paragraph{A three-symbol channel.}
Take $W_0=(0.70,0.20,0.10)$ and $W_1=(0.15,0.55,0.30)$.  Here $I_\pi$ depends on the composition and is distinct both from the channel $\chi^2$ in either direction and from the i.i.d.\ mixture Fisher constant $I_\pi^{\mathrm{mix}}=v^\top(\diag f_\pi-f_\pi f_\pi^\top)^+v$ with $f_\pi=(1-\pi)W_0+\pi W_1$.

\begin{table}[H]
\centering
\begin{tabular}{rrrrr}
\toprule
$\pi$ & $I_\pi$ & $\chi^2(W_1\Vert W_0)$ & $I_\pi^{\mathrm{mix}}$ & mixture underestimate \\
\midrule
0.2 & 1.566 & 1.445 & 1.252 & $20.0\%$ \\
0.3 & 1.635 & 1.445 & 1.217 & $25.6\%$ \\
0.5 & 1.794 & 1.445 & 1.238 & $31.0\%$ \\
0.7 & 1.988 & 1.445 & 1.402 & $29.4\%$ \\
\bottomrule
\end{tabular}
\caption{Fixed-composition constant $I_\pi$ for $W_0=(.70,.20,.10)$, $W_1=(.15,.55,.30)$: composition-dependent, and distinct from $\chi^2$ and from the i.i.d.\ mixture constant.}
\label{tab:Ipi}
\end{table}

Across compositions $\pi\in[0.05,0.95]$, $I_\pi$ ranges over $[1.47,2.30]$, a factor of $1.56$, whereas the channel $\chi^2(W_1\Vert W_0)=1.445$ is a single number.  Over the same interval the relative mixture underestimate ranges from $6.5\%$ to $31.3\%$; the four rows in \cref{tab:Ipi} show values between $20.0\%$ and $31.0\%$.  The exact Jensen--Shannon divergence of the shuffled experiment converges to $I_\pi/(8n)$, not to the mixture value.

\begin{table}[H]
\centering
\begin{tabular}{rrrr}
\toprule
$n$ & $8n\,\JSD(T_{n,k}\Vert T_{n,k+1})$ & rel.\ to $I_\pi$ & rel.\ to $I_\pi^{\mathrm{mix}}$ \\
\midrule
200 & 1.6373 & $0.15\%$ & $34.5\%$ \\
400 & 1.6361 & $0.07\%$ & $34.4\%$ \\
800 & 1.6355 & $0.04\%$ & $34.4\%$ \\
\bottomrule
\end{tabular}
\caption{Exact JSD at $\pi=0.3$ converges to $I_\pi/(8n)$ with $I_\pi=1.635$, not to the mixture value $I_\pi^{\mathrm{mix}}=1.217$.}
\label{tab:jsd-num}
\end{table}

On the local privacy scale the one-sided curve with the fixed-composition $\mu_n=\sqrt{I_\pi/n}$ tracks the exact curve, with absolute error $O(n^{-1})$ for fixed $t$ (uniformly for $t$ in compact sets), while the curve computed from the mixture covariance has a different, optimistic first-order coefficient, equivalently a different $\sqrt n$-scaled limit (it reports too small a $\delta$).

\begin{table}[H]
\centering
\begin{tabular}{rrrrr}
\toprule
$t$ & $\eps=t\mu_n$ & $\delta_{\mathrm{exact}}$ & fixed-comp.\ GDP & mixture-cov.\ GDP \\
\midrule
0.5 & 0.0226 & $8.96\cdot10^{-3}$ & $9.04\cdot10^{-3}\ (+0.9\%)$ & $6.88\cdot10^{-3}\ (-23\%)$ \\
1.0 & 0.0452 & $3.73\cdot10^{-3}$ & $3.85\cdot10^{-3}\ (+3.2\%)$ & $2.43\cdot10^{-3}\ (-35\%)$ \\
1.5 & 0.0678 & $1.27\cdot10^{-3}$ & $1.37\cdot10^{-3}\ (+7.7\%)$ & $6.71\cdot10^{-4}\ (-47\%)$ \\
2.0 & 0.0904 & $3.47\cdot10^{-4}$ & $4.02\cdot10^{-4}\ (+16\%)$ & $1.42\cdot10^{-4}\ (-59\%)$ \\
\bottomrule
\end{tabular}
\caption{One-sided privacy curve at $\pi=0.3$, $n=800$, with $\mu_n=\sqrt{I_\pi/n}=0.0452$.  The fixed-composition GDP curve tracks the exact curve, with absolute error $O(n^{-1})$ at fixed $t$ by \cref{cor:local-gdp-agreement}; the mixture-covariance GDP curve uses the wrong limiting parameter and is privacy-optimistic.}
\label{tab:curve-num}
\end{table}

\paragraph{Binary output: the composition-independent case.}
For a variance-symmetric binary channel the two covariances $\Sigma_0$ and $\Sigma_1$ coincide, so $\Sigma_\pi=\Sigma_0$ is independent of $\pi$ and $I_\pi$ is constant in $\pi$.  For $\eps_0=1$ randomized response $I_\pi=\chi^2(W_1\Vert W_0)=1.086$ for all $\pi$.  This constant is still strictly larger than the mixture constant: $I_\pi^{\rm mix}=0.854$ at $\pi=0.5$ (a $21.4\%$ underestimate), and $8n$ times the exact JSD of the shuffled binary experiment converges to $I_\pi$, not to $I_\pi^{\rm mix}$ (at $\pi=0.3$, $8n\,\JSD\to1.086$, which is $22.8\%$ above $I_\pi^{\rm mix}=0.884$).  Thus the mixture substitution is privacy-optimistic for the binary case as well; what binary symmetric channels lack is composition dependence, which is why the three-symbol example is used to display the variation of $I_\pi$ with $\pi$.  For example, the variance-asymmetric binary channel $W_0=(0.3,0.7)$, $W_1=(0.6,0.4)$ has
\[
        I_\pi=\frac{0.09}{0.21+0.03\pi},
\]
so composition dependence already occurs with two output symbols.  The local approximation is accurate here too: at $\pi=0.3$ and $t=1$, the directed curve $\delta_{T_{n,k+1}\Vert T_{n,k}}(t\mu_n)$ equals $3.83\cdot10^{-3}$ at $n=200$ against the GDP value $3.86\cdot10^{-3}$, and $1.70\cdot10^{-3}$ at $n=1000$ against $1.71\cdot10^{-3}$.

\paragraph{Comparison with worst-case amplification bounds.}
The amplification bounds of Feldman--McMillan--Talwar \cite{FeldmanMcMillanTalwar2021,FeldmanMcMillanTalwar2023} and Balle et al.\ \cite{BalleBellGasconNissim2019} answer a different question: they upper-bound $\eps$ \emph{uniformly over all neighboring datasets} for a target $\delta$.  For binary $\eps_0=1$ randomized response, \cref{tab:wc-compare} compares representative mechanism-wide upper bounds at $\delta=10^{-5}$ with the exact two-sided curve of the boundary pair $(T_{n,0},T_{n,1})$.  The FMT values are direct evaluations of the displayed pure-LDP randomized-response specializations of Theorem~3.1 in \cite{FeldmanMcMillanTalwar2021} and Theorem~3.2 in \cite{FeldmanMcMillanTalwar2023}.  The latter theorem is subject to the erratum included in the cited version; binary randomized response belongs to the corrected restricted class, so the specialization used here remains valid.  The Balle et al.\ values are obtained by numerical inversion of Theorem~5.3 in \cite{BalleBellGasconNissim2019}.  The exact column is a pair-specific benchmark, not a claim that the boundary pair is the finite-$n$ global maximizer, and none of the paper's theorems depends on the external numerical columns.  This table is diagnostic only; it should not be read as a lower bound on the true worst-case privacy profile.

\begin{table}[H]
\centering
\begin{tabular}{rrrrr}
\toprule
$n$ & exact boundary-pair $\eps$ & FMT (2021) & FMT (2023) & Balle et al. \\
\midrule
1000  & 0.105 & $0.532\ (5.0\times)$ & $0.457\ (4.3\times)$ & $0.442\ (4.2\times)$ \\
2000  & 0.071 & $0.402\ (5.6\times)$ & $0.342\ (4.8\times)$ & $0.301\ (4.2\times)$ \\
5000  & 0.043 & $0.271\ (6.4\times)$ & $0.229\ (5.4\times)$ & $0.183\ (4.3\times)$ \\
10000 & 0.029 & $0.199\ (6.9\times)$ & $0.167\ (5.8\times)$ & $0.126\ (4.4\times)$ \\
\bottomrule
\end{tabular}
\caption{Binary $\eps_0=1$ randomized response, $\delta=10^{-5}$: representative published mechanism-wide amplification upper bounds and the exact two-sided boundary-pair curve.  Ratios are relative to this pair-specific benchmark; the table is diagnostic only, does not identify the exact finite-$n$ worst-case composition, and is not used in any proof.}
\label{tab:wc-compare}
\end{table}

Relative to the exactly evaluated boundary pair, the displayed mechanism-wide bounds are larger by factors of $4.2$--$6.9$.  This comparison illustrates the difference between a pair-specific exact curve and a mechanism-wide upper bound; it is not a numerical claim about the unknown exact global maximizer.  The quantity $I_\pi$ is likewise a per-composition local constant rather than a worst-case certificate, while exact certification continues to use the accounting layer.

\section{Unbundled multi-message shuffling}\label{sec:unbundled}

In the unbundled model, each user sends $m$ independent messages and the shuffler permutes all $nm$ messages individually.  Here $m$ is fixed as $n\to\infty$.

The multi-message shuffle model has been studied extensively for protocol-level
accuracy, communication, and separation results; representative works include
Balle et al. \cite{BalleEtAlMultiMessage2020}, Ghazi et al.
\cite{GhaziEtAlMultiMessage2021}, and Girgis--Diggavi
\cite{GirgisDiggavi2024}.  The narrower purpose of this section is different:
for independent repetition of a fixed channel, it records the exact canonical
histogram likelihood ratio, its Hoeffding decomposition, and the leading
Gaussian parameter.  No protocol-level optimality or utility claim is made.

\subsection{Exact canonical likelihood ratio}

Let $N^{(m)}=(N_y)$ be the message-level histogram of all $nm$ messages.  For the canonical pair, define $w(y)=W_1(y)/W_0(y)$.

\begin{theorem}[Degree-$m$ likelihood ratio]\label{thm:unbundled-lr}
Assume $W_1\ll W_0$, so that $w=W_1/W_0$ is finite on $\supp(W_0)$; the
standing full-support assumption $\delta_{\rm full}>0$ is sufficient.  Then,
under $T^{(m)}_{n,0}$,
\[
        L^{(m)}_{n,0}(N)
        =\frac{1}{\binom{nm}{m}}
        \sum_{(m_y):\,0\le m_y\le N_y,\,\sum_y m_y=m}
        \prod_y \binom{N_y}{m_y} w(y)^{m_y}.
\]
Equivalently,
\[
        L^{(m)}_{n,0}(N)
        =\frac{[t^m]\prod_y(1+w(y)t)^{N_y}}{\binom{nm}{m}}.
\]
\end{theorem}

\begin{proof}
Under the null all $nm$ messages are drawn from $W_0$.  Under the alternative, the $m$ messages of the changed user are drawn from $W_1$ and the remaining $(n-1)m$ from $W_0$.  Conditional on the unordered histogram under the null, the changed user's $m$ positions are a uniformly random $m$-subset of the $nm$ message positions.  Averaging the product likelihood ratio over this subset gives the formula.
\end{proof}

\subsection{Leading Gaussian parameter}

\begin{proposition}[Canonical unbundled leading parameter]\label[proposition]{prop:unbundled-canonical}
Assume full support and fixed $m$.  Put $\chi^2=\chi^2(W_1\Vert W_0)$.  For the canonical unbundled pair,
\[
        L^{(m)}_{n,0}-1=\frac1n\sum_{i=1}^{nm} r(Y_i)+R^{(m)}_n,
\]
and the orthogonal decomposition below gives the exact identity
\[
        \E[(R^{(m)}_n)^2]
        =\sum_{q=2}^m\frac{\binom{m}{q}^2(\chi^2)^q}{\binom{nm}{q}}
        =O(n^{-2}).
\]
Consequently
\[
        \Var(L^{(m)}_{n,0}-1)
        =\frac{m\chi^2}{n}+O(n^{-2}),
        \qquad
        \mu_{n,\rm unb}^2(m)=\frac{m\chi^2}{n}.
\]
If $W_1\ne W_0$, then
\[
        \frac{L^{(m)}_{n,0}-1}{\sqrt{m\chi^2/n}}
        \Longrightarrow N(0,1).
\]
\end{proposition}

\begin{proof}
Put $N_*=nm$ and
\[
        h(y_1,\ldots,y_m)=\prod_{j=1}^m w(y_j).
\]
Under the null the $N_*$ messages are i.i.d.\ with law $W_0$, $\E h=1$, and \cref{thm:unbundled-lr} is the normalized $U$-statistic
\[
        U_{N_*}=\binom{N_*}{m}^{-1}
        \sum_{1\le i_1<\cdots<i_m\le N_*}
        h(Y_{i_1},\ldots,Y_{i_m}).
\]
Because $\E_{W_0}w(Y)=1$, the canonical Hoeffding projection of order $q$ is explicitly
\[
\begin{aligned}
        h_q(y_1,\ldots,y_q)
        &=\sum_{A\subseteq\{1,\ldots,q\}}(-1)^{q-|A|}
          \prod_{j\in A}w(y_j)\\
        &=\prod_{j=1}^q\{w(y_j)-1\}
          =\prod_{j=1}^q r(y_j).
\end{aligned}
\]
Thus the exact Hoeffding decomposition is
\begin{equation}\label{eq:hoeffding-unbundled}
        U_{N_*}-1
        =\sum_{q=1}^m\binom{m}{q}
          \binom{N_*}{q}^{-1}
          \sum_{1\le i_1<\cdots<i_q\le N_*}
          \prod_{j=1}^q r(Y_{i_j}).
\end{equation}
The $q=1$ term is
\[
        \frac{m}{N_*}\sum_{i=1}^{N_*}r(Y_i)
        =\frac1n\sum_{i=1}^{nm}r(Y_i).
\]
For two distinct index sets $A$ and $B$, possibly of different sizes, the product of the corresponding canonical monomials has expectation zero: an index in the symmetric difference contributes an independent factor with mean $\E r(Y)=0$.  Hence all distinct summands in \eqref{eq:hoeffding-unbundled}, including summands of different orders, are orthogonal.  Since $\E r(Y)^2=\chi^2$, the variance of the order-$q$ component is exactly
\[
        \binom{m}{q}^2\binom{N_*}{q}^{-2}
        \binom{N_*}{q}(\chi^2)^q
        =\frac{\binom{m}{q}^2(\chi^2)^q}{\binom{N_*}{q}}.
\]
Summing over $q\ge2$ proves the displayed formula for $\E[(R_n^{(m)})^2]$; because $m$ is fixed, its leading term is of order $N_*^{-2}=O(n^{-2})$.  The $q=1$ variance is $m\chi^2/n$.  Finally, after division by $\sqrt{m\chi^2/n}$, the first-order term is the standardized sum $(nm\chi^2)^{-1/2}\sum_{i=1}^{nm}r(Y_i)$, which converges to $N(0,1)$ by the central limit theorem, while the normalized remainder tends to zero in $L^2$ because its second moment is $O(n^{-1})$.  Slutsky's theorem completes the proof.
\end{proof}

We do not claim a proportional-composition multi-message extension here; proving one would require a separate posterior residual analysis analogous to \cref{prop:residual-bounds}.

\subsection{Bundled versus unbundled: leading parameter only}

In the bundled model the $m$ messages of a user are treated as a single message in alphabet $\Y^m$, and only the $n$ bundled messages are shuffled.  With $w=dW_1/dW_0$,
\[
\begin{aligned}
        1+\chi^2(W_1^{\otimes m}\Vert W_0^{\otimes m})
        &=\E_{W_0^{\otimes m}}\left[\prod_{j=1}^m w(Y_j)^2\right]\\
        &=\left(\E_{W_0}w(Y)^2\right)^m
        =(1+\chi^2(W_1\Vert W_0))^m.
\end{aligned}
\]
Hence the bundled canonical leading parameter is
\[
        \mu_{n,\rm bund}^2(m)=\frac{(1+\chi^2)^m-1}{n},
\]
whereas the unbundled parameter is $m\chi^2/n$.  Thus for $m\ge2$ and $\chi^2>0$,
\[
        \mu_{n,\rm unb}^2(m)<\mu_{n,\rm bund}^2(m).
\]
This is a comparison of leading Gaussian parameters.  It does not by itself imply a finite-$n$ dominance ordering of the exact $(\eps,\delta)$ privacy curves.  A finite-$n$ dominance claim would require either a direct positive-part comparison or explicit error bounds smaller than the parameter gap in the desired regime.

\section{Randomized-response boundary}\label{sec:boundary}

This section records how Gaussian accuracy deteriorates when the local randomized-response parameter grows with $n$.  It is included to mark the boundary of the fixed full-support regime, not as a replacement for a full non-Gaussian theory.

For binary randomized response with local parameter $\eps_0=\eps_0(n)\ge0$,
\[
        W_0(1)=q_n=(1+e^{\eps_0})^{-1},
        \qquad W_1(1)=1-q_n.
\]
Let
\[
        a_n=\frac{e^{\eps_0}}{n}.
\]
The subcritical Gaussian regime corresponds to $a_n\to0$.

\begin{proposition}[Boundary Berry--Esseen scale]\label[proposition]{prop:boundary-BE}
For the canonical shuffled randomized-response pair, assume $\eps_0(n)>0$ for all sufficiently large $n$ and $a_n\to0$.  Then
\[
        \sup_x\left|\Pp\left(\frac{\Lambda_n+\mu_n^2/2}{\mu_n}\le x\right)-\PhiG(x)\right|
        \le C\sqrt{a_n},
\]
where $\mu_n^2=\chi^2(W_1\Vert W_0)/n$.
\end{proposition}

\begin{proof}
Put $b=e^{\eps_0}$ and $q=(1+b)^{-1}$.  Under $W_0$, the canonical score $r(Y)=W_1(Y)/W_0(Y)-1$ takes the values
\[
        b-1 \quad\text{with probability }q,
        \qquad
        b^{-1}-1 \quad\text{with probability }1-q .
\]
Thus $\E r(Y)=0$ and
\[
        \Var(r(Y))=\chi^2(W_1\Vert W_0)=\frac{(b-1)^2}{b},
\]
while
\[
        \E|r(Y)|^3
        =\frac{(b-1)^3}{1+b}\left(1+\frac1{b^2}\right).
\]
Consequently
\[
        \frac{\E|r(Y)|^3}{\Var(r(Y))^{3/2}}
        =\frac{b^{3/2}(1+b^{-2})}{1+b}
        \le C b^{1/2}=C e^{\eps_0/2}.
\]
The Berry--Esseen bound for the normalized average $U_n=n^{-1}\sum_i r(Y_i)=\mu_n\xi_n$, where $\mu_n^2=\chi^2(W_1\Vert W_0)/n$, is therefore
\[
        \sup_x|\Pp(\xi_n\le x)-\PhiG(x)|
        \le C\frac{e^{\eps_0/2}}{\sqrt n}
        =C\sqrt{a_n}.
\]
Finally,
\[
        \frac{\Lambda_n+\mu_n^2/2}{\mu_n}
        =g_{\mu_n}(\xi_n),
        \qquad
        g_\mu(x)=\frac{\log(1+\mu x)+\mu^2/2}{\mu},
\]
where $g_\mu$ is increasing on the support.  Equivalently,
\[
        \Pp(g_{\mu_n}(\xi_n)\le x)=\Pp(\xi_n\le h_{\mu_n}(x)),
        \qquad
        h_\mu(x)=\frac{e^{\mu x-\mu^2/2}-1}{\mu}.
\]
We now prove the required transformation bound.  It is enough to consider $0<\mu\le1/2$.  If $|\mu x|\le1/2$, put $y=\mu x-\mu^2/2$.  Taylor's formula gives
\begin{equation}\label{eq:h-transform-local}
        |h_\mu(x)-x|\le C\mu(1+x^2).
\end{equation}
Moreover, on the fixed interval containing all such $y$, the function
\[
        q(y)=\begin{cases}(e^y-1)/y,&y\ne0,\\1,&y=0,\end{cases}
\]
is bounded above and below by positive constants, and
\[
        h_\mu(x)=q(y)(x-\mu/2).
\]
If $|x|\le2\mu$, \eqref{eq:h-transform-local} and the mean-value theorem give $|\PhiG(h_\mu(x))-\PhiG(x)|\le C\mu$.  If $|x|>2\mu$, then $x$ and $h_\mu(x)$ have the same sign and $|h_\mu(x)|\ge c|x|$.  Every point between $x$ and $h_\mu(x)$ therefore has absolute value at least $c'|x|$, and the mean-value theorem together with \eqref{eq:h-transform-local} gives
\[
        |\PhiG(h_\mu(x))-\PhiG(x)|
        \le C\mu(1+x^2)e^{-c x^2}\le C\mu
        \qquad (|\mu x|\le1/2).
\]
If $x>1/(2\mu)$, then both $x$ and $h_\mu(x)$ are at least $c/\mu$; if $x<-1/(2\mu)$, both are at most $-c/\mu$.  Gaussian tail bounds then give
\[
        |\PhiG(h_\mu(x))-\PhiG(x)|\le Ce^{-c/\mu^2}\le C\mu.
\]
Thus
\[
        \sup_x|\PhiG(h_\mu(x))-\PhiG(x)|\le C\mu .
\]
Since $\mu_n^2=(b-1)^2/(bn)\le b/n=a_n$, this additional transformation error is $O(\sqrt{a_n})$.  Combining the two bounds proves the claim.
\end{proof}

If $\eps_0(n)=0$, then $W_0=W_1$, the likelihood ratio is identically one, and the zero-signal case is trivial rather than a normalized limit theorem.  The proposition makes no claim when $a_n$ does not tend to zero.  Critical and supercritical scalings are outside the scope of this manuscript; their Poisson, Skellam, compound-Poisson, and hybrid limits are treated in Parts~II--III \cite{ShvetsPartII2026,ShvetsPartIII2026}.

\section{Discussion}\label{sec:discussion}

The central object in the full-support fixed-composition shuffle regime is the Fisher constant
\[
        I_\pi=v^\top\Sigma_\pi^+v,
        \qquad \Sigma_\pi=(1-\pi)\Sigma_0+\pi\Sigma_1.
\]
It governs the leading JSD constant, smooth $f$-divergence constants, the normalized privacy-loss expansion, and both directed local privacy curves.  The fixed-composition covariance is essential; replacing it by the i.i.d. mixture covariance changes the Fisher constant in the optimistic direction.

The paper separates three layers.  The exact layer consists of likelihood-ratio and positive-part identities and can be used for finite-$n$ certification.  The asymptotic information-geometric layer identifies $I_\pi$ and the Gaussian tangent approximation.  The local privacy layer gives matching expansions for both orderings at $\eps=t\sqrt{I_\pi/n}$.  A global GDP approximation with an $O(n^{-1/2})$ error is only a coarse tangent statement and should not be interpreted as a finite-$n$ certificate. The contribution of the paper is the identification of this constant for fixed-composition neighboring shuffled histogram laws and of the local geometry it controls, not a tighter worst-case certificate: exact certification remains the task of the accounting layer, and the same exact curve is obtained by positive-part or privacy-loss-distribution methods.

Companion work already treats canonical growing-alphabet experiments
\cite{ShvetsGrowing2026} and the non-Gaussian critical and hybrid regimes
\cite{ShvetsPartII2026,ShvetsPartIII2026}.  Remaining open directions include
noncanonical fixed-composition theory with growing alphabets, second-order
Edgeworth corrections for the local privacy curve, and direct finite-$n$
dominance comparisons for unbundled versus bundled shuffling.

\appendix

\section{Lattice ratio lemmas}\label{app:edgeworth}

This appendix proves the local ratio estimates used in \cref{sec:linearization}.  The load-bearing point is the derivative control on the relative second-order remainder, because that derivative produces the bounded-shift estimate used in \cref{prop:residual-bounds}.  The argument below carries out the exponential tilt, the three-region Fourier decomposition, the cumulant expansion, and the differentiation of the remainder.  The classical references Bhattacharya--Rao \cite{BhattacharyaRao2010}, Petrov \cite[Ch.~VII]{Petrov1975}, and Kolassa \cite{Kolassa2006} are cited for background on local Edgeworth and saddlepoint methods; no unquoted theorem from those sources is used to supply the derivative bound stated below.

For orientation, the appendix has two outputs.  \Cref{lem:ratio} gives the first-order bounded-shift log-ratio used in the pointwise conditional-expectation linearization.  \Cref{lem:refined-ratio} isolates the order-$m^{-1}$ Edgeworth-gradient term and controls the remaining error by $m^{-3/2}(1+\|z\|^9)$; this is the estimate used in the $L^2/L^4$ residual theorem.

\paragraph{Uniformity convention.}
All constants in this appendix are uniform over all triangular arrays satisfying the full-support assumption \emph{(A.0)}, over all lattice points and real extensions in the displayed moderate-deviation windows, and over all bounded tangent shifts used in the ratio lemmas.  In the application to \cref{prop:residual-bounds}, the same constants are also uniform over all compositions $k/n\in[\eta,1-\eta]$.  Constants may depend on displayed fixed parameters such as $d,p,A_0,A,B$, and $\eta$, but not on $m$, $s$, the particular array, or the particular composition.
Whenever a statement below holds for all sufficiently large $m$, its
threshold $m_0$ is uniform over the arrays and depends only on the displayed
fixed parameters.

Let $X_1,\ldots,X_m$ be independent one-hot random vectors in $\{e_1,\ldots,e_d\}$, where $d$ is fixed.  Assume a uniform full-support condition: for some $p>0$,
\[
        \Pp(X_i=e_y)\ge p,
        \qquad 1\le i\le m,\quad 1\le y\le d .
        \tag{A.0}
\]
Let
\[
        S_m=\sum_{i=1}^m X_i,
        \qquad \mu=\E S_m,
        \qquad \Sigma_{\rm ps}=m^{-1}\Cov(S_m).
\]
Since $\one^\top S_m=m$, $\Sigma_{\rm ps}$ is singular with kernel $\operatorname{span}\{\one\}$ and is nonsingular on $\mathcal T=\{x:\one^\top x=0\}$.  We write $x^-=(x_1,\ldots,x_{d-1})$ for the vector obtained by deleting the last coordinate, and $\Sigma_{\rm ps}^-$ for the corresponding covariance minor.  Uniform full support implies that the eigenvalues of $\Sigma_{\rm ps}^-$ are bounded above and away from zero uniformly in $m$; this is proved in the next lemma in the tilted form needed below.

Put $r=d-1$ and $Y_i=X_i^-\in A:=\{0,e_1,\ldots,e_r\}\subset\mathbb Z^r$.  For $\theta\in\mathbb R^r$ define
\[
        K_i(\theta)=\log\E e^{\theta^\top Y_i},
        \qquad
        K_m(\theta)=\frac1m\sum_{i=1}^mK_i(\theta).
\]
Thus $\nabla K_m(0)=\mu^-/m$ and $\nabla^2K_m(0)=\Sigma_{\rm ps}^-$.  Expectations under the exponential tilt with parameter $\theta$ are denoted by $\E_\theta$.

\begin{lemma}[Uniform analytic control under tilting]\label[lemma]{lem:tilt-control}
There are constants $\theta_0,\lambda>0$ and $C_q<\infty$, depending only on $(d,p)$, such that for every $m$, every array satisfying \emph{(A.0)}, and every $\|\theta\|_2\le\theta_0$,
\[
        \lambda I_r\preceq\nabla^2K_m(\theta)\preceq\lambda^{-1}I_r,
        \qquad
        \|D^qK_m(\theta)\|\le C_q,
        \quad 3\le q\le8.
        \tag{A.1a}
\]
Moreover, for each fixed $A_0<\infty$ there is
$m_0=m_0(d,p,A_0)$ such that, for every $m\ge m_0$, every
$z$ with $\|z\|_2\le A_0m^{1/8}$ determines a unique
$\theta_m(z)$ in $\|\theta\|_2\le\theta_0$ through
\[
        \nabla K_m(\theta_m(z))=\nabla K_m(0)+m^{-1/2}z.
\]
This map is continuously differentiable and satisfies, uniformly on the stated window,
\[
\begin{aligned}
        \theta_m(z)
        &=m^{-1/2}(\Sigma_{\rm ps}^-)^{-1}z
          +O\!\left(m^{-1}\|z\|_2^2\right),\\
        \|\theta_m(z)\|_2&=O(m^{-3/8}),\\
        D_z\theta_m(z)
        &=m^{-1/2}\{\nabla^2K_m(\theta_m(z))\}^{-1},
        \qquad \|D_z\theta_m(z)\|=O(m^{-1/2}).
\end{aligned}
        \tag{A.1b}
\]
\end{lemma}

\begin{proof}
Let $q_i(a)=\Pp(Y_i=a)$.  Under \emph{(A.0)}, $q_i(a)\ge p$ for every $a\in A$.  The tilted probabilities are
\[
        q_{i,\theta}(a)
        =\frac{q_i(a)e^{\theta^\top a}}
        {\sum_{b\in A}q_i(b)e^{\theta^\top b}}.
\]
For $\|\theta\|_2\le\theta_0$, all $a\in A$ satisfy $|\theta^\top a|\le\theta_0$, hence
\[
        q_{i,\theta}(a)\ge p e^{-2\theta_0}=:p_0>0.
\]
For a unit vector $u\in\mathbb R^r$ and $V=u^\top Y_i$, the pairwise variance identity gives
\[
\begin{aligned}
        \Var_\theta(V)
        &=\frac12\sum_{a,b\in A}q_{i,\theta}(a)q_{i,\theta}(b)
              \{u^\top(a-b)\}^2\\
        &\ge \sum_{j=1}^r q_{i,\theta}(0)q_{i,\theta}(e_j)u_j^2
        \ge p_0^2.
\end{aligned}
\]
Also $\Var_\theta(V)\le\E_\theta V^2\le1$.  Thus every tilted covariance has eigenvalues in $[p_0^2,1]$, and the same is true of their average $\nabla^2K_m(\theta)$.  Derivatives of $K_i$ of order $q\ge2$ are joint cumulants of coordinates of the bounded vector $Y_i$ under the tilted law.  Since $\|Y_i\|_2\le1$, these cumulants are bounded by constants depending only on $q$ and $r$; averaging proves the derivative bounds.

Set $H_m=\nabla^2K_m(0)=\Sigma_{\rm ps}^-$.  Write
\[
        \nabla K_m(\theta)-\nabla K_m(0)=H_m\theta+R_m(\theta).
\]
The third-derivative bound gives, on $\|\theta\|\le\theta_0$,
\[
        \|R_m(\theta)\|\le C\|\theta\|^2,
        \qquad
        \|R_m(\theta)-R_m(\theta')\|
        \le C(\|\theta\|+\|\theta'\|)\|\theta-\theta'\|.
\]
For $\delta=m^{-1/2}z$, put $r_\delta=2\lambda^{-1}\|\delta\|$ and define
\[
        \mathcal F_\delta(\theta)=H_m^{-1}\{\delta-R_m(\theta)\}.
\]
Because $\|\delta\|=O(m^{-3/8})$, for all sufficiently large $m$ the ball $\{\|\theta\|\le r_\delta\}$ lies inside $\{\|\theta\|\le\theta_0\}$, $\mathcal F_\delta$ maps this ball into itself, and its Lipschitz constant is at most $1/2$.  Banach's fixed-point theorem therefore supplies a unique solution in this ball.  The lower Hessian bound makes $\nabla K_m$ strongly monotone on the whole $\theta_0$-ball, so no second solution can occur there.  The fixed-point equation gives
\[
        \theta_m(z)=H_m^{-1}\delta+O(\|\delta\|^2),
        \qquad \|\theta_m(z)\|=O(\|\delta\|),
\]
which proves the first two lines of \emph{(A.1b)}.  The implicit-function theorem applies because the Hessian is invertible; differentiating the defining equation gives the displayed formula for $D_z\theta_m(z)$, and the inverse-Hessian bound completes the proof.
\end{proof}

\begin{lemma}[Uniform tilted aperiodicity]\label[lemma]{lem:tilted-aperiodicity}
Let $A=\{0,e_1,\ldots,e_{d-1}\}\subset\Z^{d-1}$.  For every $q_0>0$ and every $\delta>0$ there is $\rho<1$, depending only on $(d,q_0,\delta)$, such that the following holds.  If $q$ is a probability vector on $A$ with $q(a)\ge q_0$ for all $a\in A$, then
\[
        \sup_{u\in[-\pi,\pi]^{d-1}:\ \|u\|_2\ge\delta}
        \left|\sum_{a\in A}q(a)e^{iu^\top a}\right|\le\rho .
\]
The same bound holds for centered variables, since centering multiplies the characteristic function by a phase.
\end{lemma}

\begin{proof}
The set of probability vectors with all coordinates at least $q_0$ is compact.  The displayed modulus is continuous in $(q,u)$.  If its value were one, all phases $e^{iu^\top a}$ on the support $A$ would be equal.  Since $0,e_1,\ldots,e_{d-1}$ are all in the support, this forces each coordinate of $u$ to be an integer multiple of $2\pi$.  In $[-\pi,\pi]^{d-1}$ this gives only $u=0$, which is excluded by $\|u\|_2\ge\delta$.  Therefore the compact maximum is strictly smaller than one.
\end{proof}

\begin{lemma}[Second-order local Edgeworth expansion]\label[lemma]{lem:edgeworth}
Fix $A_0<\infty$.  There exist
$m_0=m_0(d,p,A_0)$ and $C<\infty$, depending only on
$(d,p,A_0)$, such that the following holds.  For every $m\ge m_0$
and every row $(X_1,\ldots,X_m)$ satisfying \emph{(A.0)}, there are
polynomials $Q_m^{(1)}$ and $Q_m^{(2)}$ on $\R^{d-1}$, of degrees at
most $3$ and $6$, respectively, whose coefficients are bounded uniformly
over all such rows.  The polynomials themselves may depend on the row.
For every $s\in\Z_{\ge0}^d$ with $\one^\top s=m$ and
\[
        z^-:=\frac{s^- -\mu^-}{\sqrt m},
        \qquad \|z^-\|_2\le A_0m^{1/8},
\]
we have
\begin{equation}
\begin{aligned}
\Pp(S_m=s)
&=\frac{1}{(2\pi m)^{(d-1)/2}\sqrt{\det\Sigma_{\rm ps}^-}}
\exp\left(-\frac12(z^-)^\top(\Sigma_{\rm ps}^-)^{-1}z^-\right)\\
&\quad\times
\left(1+m^{-1/2}Q_m^{(1)}(z^-)+m^{-1}Q_m^{(2)}(z^-)+R_m^{(2)}(z^-)\right).
\end{aligned}
\tag{A.1}
\end{equation}
where
\[
\begin{aligned}
        |Q_m^{(1)}(z)|&\le C(1+\|z\|_2^3),
        &\|\nabla Q_m^{(1)}(z)\|&\le C(1+\|z\|_2^2),\\
        |Q_m^{(2)}(z)|&\le C(1+\|z\|_2^6),
        &\|\nabla Q_m^{(2)}(z)\|&\le C(1+\|z\|_2^5),\\
        |R_m^{(2)}(z)|&\le C m^{-3/2}(1+\|z\|_2^9),
        &\|\nabla R_m^{(2)}(z)\|&\le C m^{-3/2}(1+\|z\|_2^8).
\end{aligned}
\tag{A.2}
\]
Moreover, for every fixed $B<\infty$, whenever $\|h\|_2\le Bm^{-1/2}$ and both $z$ and $z+h$ lie in the window $\|\cdot\|_2\le A_0m^{1/8}$,
\[
        |R_m^{(2)}(z+h)-R_m^{(2)}(z)|
        \le C_B m^{-2}(1+\|z\|_2^8).
        \tag{A.3}
\]
\end{lemma}

\begin{proof}
We prove the relative form by an exponential-tilting Fourier argument.  The tilt is used to obtain a relative expansion in the moderate window; an absolute local Edgeworth expansion would not suffice after division by the Gaussian density.

Use the notation introduced before \cref{lem:tilt-control}.  Fix a lattice point $s$ in the stated window, set
\[
        z=\frac{s^- -\mu^-}{\sqrt m},
        \qquad
        \theta=\theta_m(z),
\]
and note from \cref{lem:tilt-control} that
\[
        \nabla K_m(\theta)=s^-/m,
        \qquad
        \|\theta\|=O(m^{-3/8}),
        \qquad
        D_z\theta=O(m^{-1/2}).
\]
For later differentiation we use the same construction for every real $z$ in the window, not only for lattice values.  Define
\[
\begin{aligned}
        \mathcal J_m(z)
        :=\frac{1}{(2\pi)^r}\int_{[-\pi,\pi]^r}
        \prod_{i=1}^m
        \E_{\theta_m(z)}
        e^{iu^\top(Y_i-\E_{\theta_m(z)}Y_i)}\,du .
\end{aligned}
        \tag{A.1c0}
\]
The integrand is continuously differentiable in $z$, and the domain is compact, so $\mathcal J_m$ is a $C^1$ function.  Conjugate symmetry under $u\mapsto-u$ makes $\mathcal J_m(z)$ real.  When $z=(s^- -\mu^-)/\sqrt m$ for a lattice point $s$, the tilted sum has mean $s^-$ and Fourier inversion gives
$\mathcal J_m(z)=\Pp_{\theta_m(z)}(S_m=s)$.  Thus \emph{(A.1c0)} supplies a definite smooth extension of the lattice quantity; the remainder $R_m^{(2)}$ below is defined through this extension.  This makes the gradient and mean-value arguments literal rather than formal.

Introduce the exponential tilt
\[
        \frac{d\mathbb P_\theta}{d\mathbb P}(Y_1,\ldots,Y_m)
        =\exp\left\{\theta^\top\sum_iY_i-\sum_iK_i(\theta)\right\}.
\]
Under $\mathbb P_\theta$, the sum $\sum_iY_i$ has mean $s^-$ and covariance $m\Sigma_\theta$, where $\Sigma_\theta=\nabla^2K_m(\theta)$.  Thus
\[
        \mathbb P(S_m=s)
        =\exp\left\{mK_m(\theta)-\theta^\top s^-\right\}
          \mathbb P_\theta(S_m=s).
        \tag{A.1c}
\]
Fourier inversion under $\mathbb P_\theta$ gives
\[
\mathbb P_\theta(S_m=s)
=\frac{1}{(2\pi)^r}
  \int_{[-\pi,\pi]^r}
  \prod_{i=1}^m
  \mathbb E_\theta e^{iu^\top(Y_i-\mathbb E_\theta Y_i)}
  \,du .
\]
After the change $u=w/\sqrt m$, split the domain into the central region $\|w\|\le m^{1/12}$, the intermediate annulus $m^{1/12}<\|w\|\le\delta\sqrt m$, and the outer region.  Choose $\delta>0$ small but fixed.  The tilted atom probabilities are bounded below by $p_0>0$.  If $\varphi_{i,\theta}(u)=\E_\theta e^{iu^\top Y_i}$ and $Y_i'$ is an independent copy, then
\[
\begin{aligned}
        1-|\varphi_{i,\theta}(u)|^2
        &=\E_\theta\{1-\cos(u^\top(Y_i-Y_i'))\}\\
        &\ge 2p_0^2\sum_{j=1}^r(1-\cos u_j)
        \ge c\|u\|^2
\end{aligned}
\]
for $\|u\|\le\delta$, after choosing a fixed $\delta\le1$.  Centering changes only the phase.  Hence, uniformly in $i,m,\theta$,
\[
        \left|\E_\theta e^{iu^\top(Y_i-\E_\theta Y_i)}\right|
        \le e^{-c\|u\|^2},
        \qquad \|u\|\le\delta.
\]
Consequently the product on the intermediate annulus is at most $e^{-c\|w\|^2}$ and its integral is $O(e^{-cm^{1/6}})$.  The derivative of one centered characteristic factor with respect to $\theta$ is $O(\|u\|)$ for $\|u\|\le\delta$: it vanishes at $u=0$, and its first $u$-derivative is uniformly bounded because the tilted variables are bounded and their first two tilted moments have uniformly bounded $\theta$-derivatives.  The product rule therefore gives
\[
        \left\|D_\theta\prod_{i=1}^m
        \E_\theta e^{iu^\top(Y_i-\E_\theta Y_i)}\right\|
        \le C m\|u\|e^{-cm\|u\|^2}.
\]
With $u=w/\sqrt m$ and $D_z\theta_m(z)=O(m^{-1/2})$, the differentiated intermediate integrand is bounded by $C\|w\|e^{-c\|w\|^2}$, whose integral over $\|w\|>m^{1/12}$ is $O(e^{-cm^{1/6}})$.  On the outer region, \cref{lem:tilted-aperiodicity} applies to each tilted summand: because $\theta_m(z)=O(m^{-3/8})$, all tilted atom probabilities are bounded below by a constant depending only on $(d,p)$.  The undifferentiated product is $O(\rho^m)$ and its $\theta$-derivative is $O(m\rho^{m-1})$ for some $\rho<1$; after multiplication by $D_z\theta_m(z)$ and integration over the bounded Fourier cube, both are $O(e^{-cm})$.

On the central region, write $\bar\kappa_{j,\theta}$ for the averaged centered cumulant tensor of order $j$ under $\mathbb P_\theta$.  For $t=w/\sqrt m$, the characteristic function is uniformly within a fixed neighborhood of one; we use the analytic logarithm that equals zero at $t=0$.  For each $i$ the centered tilted cumulant generating function has the exact representation
\[
        g_i(\theta,t):=\log\E_\theta e^{it^\top(Y_i-\E_\theta Y_i)}
        =K_i(\theta+it)-K_i(\theta)-it^\top\nabla K_i(\theta),
\]
where $K_i=\log M_i$ and $M_i(z)=\E e^{z^\top Y_i}$.  The moment function $M_i$ is entire (a finite sum of exponentials), but $K_i$ is analytic only where $M_i\ne0$, so we first record a zero-free bound on a complex ball about $\theta$.  For complex $h\in\mathbb C^r$, since $\|Y_i\|_2\le1$ gives $|h^\top Y_i|\le\|h\|_2$,
\[
        \left|\frac{M_i(\theta+h)}{M_i(\theta)}-1\right|
        =\bigl|\E_\theta e^{h^\top Y_i}-1\bigr|
        \le\E_\theta\bigl|e^{h^\top Y_i}-1\bigr|
        \le e^{\|h\|_2}-1,
\]
so for $\|h\|_2\le r_0:=\log(3/2)$ the right-hand side is at most $\tfrac12$ and $M_i(\theta+h)\ne0$.  Thus $M_i$ is zero-free on the complex ball $\{z:\|z-\theta\|_2\le r_0\}$, and the analytic branch of $\log M_i$ that agrees with the real-valued $K_i$ at $z=\theta$ is defined on it; the representation holds for $\|\theta\|\le\theta_0$ and $\|t\|_2\le\delta$ with any fixed $\delta<r_0$.  Expanding in $t$ gives $g_i(\theta,t)=\sum_{j\ge2}(i^j/j!)\,\kappa_{j,i,\theta}[t^{\otimes j}]$, where $\kappa_{j,i,\theta}=D^jK_i(\theta)$ is the order-$j$ tilted cumulant tensor of $Y_i$.  The proof of \cref{lem:tilt-control} bounds the individual cumulants of each $Y_i$, not only their average: since every tilted atom probability is at least $p_0>0$ and $\|Y_i\|_2\le1$, the entire $M_i$ satisfies $c_1\le|M_i(z)|\le c_2$ on $\{\|z-\theta\|_2\le r_0\}$ with $c_1,c_2$ depending only on $(\theta_0,r_0,r)$; since each evaluation point $\theta+it$ lies at distance at most $\delta<r_0$ from $\theta$, Cauchy estimates on balls of fixed radius $(r_0-\delta)/2$ give $\|D^qK_i(\theta+it)\|\le C_q$ for $q\le7$, uniformly in $i$, $m$, and the array.  One further $\theta$-derivative of $\kappa_{j,i,\theta}$ is a derivative of $K_i$ of order $j+1$; the order-$5$ Taylor remainder of $g_i$ in $t$ uses $\kappa_{6,i,\theta}$, and one $\theta$-derivative of it uses $D^7K_i$.  Hence, for $|\beta|\le1$,
\[
\left|D_\theta^\beta\left[
\log\E_\theta e^{it^\top(Y_i-\E_\theta Y_i)}
+\frac12\kappa_{2,i,\theta}[t,t]
-\sum_{j=3}^5\frac{i^j}{j!}\kappa_{j,i,\theta}[t^{\otimes j}]
\right]\right|
\le C\|t\|^6.
\]
Summing this inequality over $i$ yields the cumulant expansion
\[
\begin{aligned}
\sum_{i=1}^m\log\mathbb E_\theta e^{i(w/\sqrt m)^\top(Y_i-\mathbb E_\theta Y_i)}
&=-\tfrac12 w^\top\Sigma_\theta w
  +m^{-1/2}C_{3,m,\theta}(w)+m^{-1}C_{4,m,\theta}(w) \\
&\quad +m^{-3/2}C_{5,m,\theta}(w)+O\!\left(m^{-2}(1+\|w\|^6)\right),
\end{aligned}
\]
where $C_{j,m,\theta}$ is homogeneous of degree $j$ up to the factor $i^j/j!$, and all coefficients and one $\theta$-derivative are uniformly bounded.  To make the exponentiation remainder explicit, write
\[
\begin{aligned}
        G_{m,\theta}(w)
        &=m^{-1/2}C_{3,m,\theta}(w)
          +m^{-1}C_{4,m,\theta}(w)
          +m^{-3/2}C_{5,m,\theta}(w)
          +\Delta_{m,\theta}(w),
\end{aligned}
\]
where, for $|\beta|\le1$,
\[
        |D_\theta^\beta\Delta_{m,\theta}(w)|
        \le Cm^{-2}(1+\|w\|^6).
\]
On $\|w\|\le m^{1/12}$, $|G_{m,\theta}(w)|\le Cm^{-1/4}$, uniformly in the array and in $\theta$.  Taylor's formula with integral remainder for $e^x$, applied also after one $\theta$-derivative, therefore gives
\[
\begin{aligned}
 e^{G_{m,\theta}(w)}
 &=1+m^{-1/2}C_{3,m,\theta}(w)\\
 &\quad+m^{-1}\left\{C_{4,m,\theta}(w)
              +\frac12C_{3,m,\theta}(w)^2\right\}\\
 &\quad+m^{-3/2}\left\{C_{5,m,\theta}(w)
              +C_{3,m,\theta}(w)C_{4,m,\theta}(w)
              +\frac16C_{3,m,\theta}(w)^3\right\}
   +m^{-2}\widetilde{\mathcal B}_{m,\theta}(w),
\end{aligned}
\]
with
\[
        |D_\theta^\beta\widetilde{\mathcal B}_{m,\theta}(w)|
        \le C(1+\|w\|^{12}),
        \qquad |\beta|\le1.
\]
Multiplication by the Gaussian factor and differentiation of that factor add at most two powers of $\|w\|$.  Consequently
\[
\begin{aligned}
\prod_i \mathbb E_\theta e^{i(w/\sqrt m)^\top(Y_i-\mathbb E_\theta Y_i)}
&=e^{-\frac12w^\top\Sigma_\theta w}
\bigl(1+m^{-1/2}A_{1,m,\theta}(w)+m^{-1}A_{2,m,\theta}(w)\\
&\qquad\qquad +m^{-3/2}A_{3,m,\theta}(w)\bigr)
  +m^{-2}\mathcal B_{m,\theta}(w),
\end{aligned}
        \tag{A.1d}
\]
where $A_{1,m,\theta}$, $A_{2,m,\theta}$, and $A_{3,m,\theta}$ are polynomials of degrees at most $3$, $6$, and $9$, respectively, and
\[
        |D_\theta^\beta \mathcal B_{m,\theta}(w)|
        \le C(1+\|w\|^{16})e^{-c\|w\|^2},
        \qquad |\beta|\le1.
        \tag{A.1e}
\]
\paragraph{$C^1$ Fourier integration.}
We spell out the integration step, including the derivative of the Fourier
remainder.  Put
\[
        \mathcal C_m=\{w\in\R^r:\|w\|_2\le m^{1/12}\},
        \qquad
        \mathcal D_m=[-\pi\sqrt m,\pi\sqrt m]^r,
\]
and
\[
        \Phi_{m,\theta}(w)
        :=\prod_{i=1}^m
        \E_\theta
        e^{i(w/\sqrt m)^\top(Y_i-\E_\theta Y_i)}.
\]
After the change of variables used above,
\[
        \mathcal J_m(z)
        =\frac{m^{-r/2}}{(2\pi)^r}
          \int_{\mathcal D_m}\Phi_{m,\theta}(w)\,dw,
        \qquad \theta=\theta_m(z).
\]
The intermediate-region estimate before \emph{(A.1d)} and its
$\theta$-derivative are exponentially small after absorbing the polynomial
factor $\sqrt m$ into the exponential bound.  The outer-region contribution,
including its $\theta$-derivative, is $O(m^{r/2+1}\rho^m)$.  Consequently,
after decreasing $c$,
\[
\left|\int_{\mathcal D_m\setminus\mathcal C_m}
        \Phi_{m,\theta}(w)\,dw\right|
+\left\|D_\theta
        \int_{\mathcal D_m\setminus\mathcal C_m}
        \Phi_{m,\theta}(w)\,dw\right\|
\le Ce^{-c m^{1/6}}.
\]
Here and below differentiation under the integral is legitimate because
the integration domain is bounded and every tilted characteristic factor
is analytic in $\theta$.

Write
\[
\begin{aligned}
        \mathcal P_{m,\theta}(w)
        :=e^{-\frac12w^\top\Sigma_\theta w}
        \bigl(&1+m^{-1/2}A_{1,m,\theta}(w)
                 +m^{-1}A_{2,m,\theta}(w)\\
              &+m^{-3/2}A_{3,m,\theta}(w)\bigr).
\end{aligned}
\]
The coefficients of these polynomials and their first $\theta$-derivatives
are uniformly bounded by the cumulant estimates above.  Uniform
ellipticity of $\Sigma_\theta$, together with
$\|D_\theta\Sigma_\theta\|\le C$, therefore yields
\[
\left|\int_{\R^r\setminus\mathcal C_m}
        \mathcal P_{m,\theta}(w)\,dw\right|
+\left\|D_\theta
        \int_{\R^r\setminus\mathcal C_m}
        \mathcal P_{m,\theta}(w)\,dw\right\|
\le Ce^{-c m^{1/6}}.
\]
Moreover, \emph{(A.1e)} gives directly
\[
\left|\int_{\mathcal C_m}\mathcal B_{m,\theta}(w)\,dw\right|
+\left\|D_\theta
        \int_{\mathcal C_m}\mathcal B_{m,\theta}(w)\,dw\right\|
\le C.
\]
The polynomials $A_{1,m,\theta}$ and $A_{3,m,\theta}$ are odd in $w$,
whereas the Gaussian factor and $\mathcal C_m$ are even.  Their central
integrals, and also the $\theta$-derivatives of those integrals, thus
vanish identically.  Combining this parity cancellation with
\emph{(A.1d)} and the three bounds above gives
\[
\begin{aligned}
 \int_{\mathcal D_m}\Phi_{m,\theta}(w)\,dw
 &=\int_{\R^r}e^{-\frac12w^\top\Sigma_\theta w}
       \bigl\{1+m^{-1}A_{2,m,\theta}(w)\bigr\}\,dw
       +m^{-2}\mathcal R_m^{\rm F}(\theta),
\end{aligned}
\]
where
\[
        |\mathcal R_m^{\rm F}(\theta)|
        +\|D_\theta\mathcal R_m^{\rm F}(\theta)\|\le C.
\]
This is a $C^1$ identity uniformly over the arrays and all relevant
tilted parameters.  Define
\[
        B_m(\theta)
        :=\frac{\sqrt{\det\Sigma_\theta}}{(2\pi)^{r/2}}
          \int_{\R^r}e^{-\frac12w^\top\Sigma_\theta w}
             A_{2,m,\theta}(w)\,dw .
\]
Gaussian moment bounds, uniform ellipticity, and the coefficient bounds
give
\[
        |B_m(\theta)|+\|D_\theta B_m(\theta)\|\le C.
\]
After absorbing the uniformly bounded Gaussian normalization into
$\mathcal R_m^{\rm F}$ and setting
$E_m^{\rm F}(\theta)=m^{-2}\mathcal R_m^{\rm F}(\theta)$, we obtain
\[
        \mathcal J_m(z)
        =\frac{1}{(2\pi m)^{r/2}\sqrt{\det\Sigma_\theta}}
          \left\{1+m^{-1}B_m(\theta)+E_m^{\rm F}(\theta)\right\},
        \tag{A.1f}
\]
with
\[
        |E_m^{\rm F}(\theta)|
        +\|D_\theta E_m^{\rm F}(\theta)\|
        \le Cm^{-2}.
\]
In particular, \emph{(A.1b)} and the chain rule also give
\[
        \left\|\nabla_z
        E_m^{\rm F}(\theta_m(z))\right\|\le Cm^{-5/2}.
\]
For lattice $z$, \emph{(A.1f)} is the corresponding formula for
$\Pp_\theta(S_m=s)$.  Thus both the Fourier remainder and its derivative
are controlled before the saddlepoint factors are expanded.

It remains to expand the saddlepoint factors.  Put $H_m=\Sigma_{\rm ps}^-$, $\delta=m^{-1/2}z$, and write $\vartheta_m(\delta)$ for the solution of
$\nabla K_m(\theta)=\nabla K_m(0)+\delta$.  Define the local Legendre function
\[
        \Psi_m(\delta)
        :=K_m(\vartheta_m(\delta))
          -\vartheta_m(\delta)^\top\{\nabla K_m(0)+\delta\}.
\]
The envelope identity gives $\nabla_\delta\Psi_m(\delta)=-\vartheta_m(\delta)$; hence
$D^2\Psi_m(0)=-H_m^{-1}$.  The map $\delta\mapsto\vartheta_m(\delta)$ defined by $\nabla K_m(\vartheta_m(\delta))=\nabla K_m(0)+\delta$ is differentiated implicitly: writing $H_\theta=\nabla^2K_m(\vartheta_m(\delta))$,
\[
        D_\delta\vartheta_m=H_\theta^{-1},
        \qquad
        D_\delta^2\vartheta_m[u,v]
        =-H_\theta^{-1}\,D^3K_m(\vartheta_m)\bigl[H_\theta^{-1}u,\,H_\theta^{-1}v\bigr],
\]
and, inductively, $D_\delta^j\vartheta_m$ is a finite sum of products of $H_\theta^{-1}$ and derivatives $D^\ell K_m(\vartheta_m)$ with $\ell\le j+1$.  By \cref{lem:tilt-control} the inverse Hessian is uniformly bounded and $\|D^\ell K_m\|\le C_\ell$ for $\ell\le8$, so $D_\delta^j\vartheta_m$ is uniformly bounded for $j\le4$ on the relevant neighborhood.  Since the envelope identity gives $\nabla_\delta\Psi_m(\delta)=-\vartheta_m(\delta)$, this yields uniform bounds on $D^j\Psi_m$ for $3\le j\le5$, and the same products bound the derivatives of the determinant factor below.  Applying Taylor's theorem with integral remainder separately to $\Psi_m$ and to $\nabla\Psi_m$ therefore yields
\[
\begin{aligned}
        mK_m(\theta)-\theta^\top s^-
        &=m\Psi_m(m^{-1/2}z)\\
        &=-\frac12 z^\top H_m^{-1}z
          +m^{-1/2}P_{3,m}(z)+m^{-1}P_{4,m}(z)
          +m^{-3/2}E_{5,m}(z),
\end{aligned}
        \tag{A.1g}
\]
where $P_{3,m}$ and $P_{4,m}$ have degrees at most $3$ and $4$ and uniformly bounded coefficients, while
\[
        |E_{5,m}(z)|\le C(1+\|z\|^5),
        \qquad
        \|\nabla E_{5,m}(z)\|\le C(1+\|z\|^4).
\]

For the determinant term, the function
\[
        \Gamma_m(\delta)
        :=\frac12\log\frac{\det H_m}
        {\det\nabla^2K_m(\vartheta_m(\delta))}
\]
has uniformly bounded derivatives through order three.  Writing $A(\delta)=\nabla^2K_m(\vartheta_m(\delta))$, so that $\Gamma_m(\delta)=\tfrac12\log\det H_m-\tfrac12\log\det A(\delta)$, Jacobi's formula gives
\[
        D\Gamma_m[u]=-\tfrac12\operatorname{tr}\!\bigl(A^{-1}DA[u]\bigr),
\]
\[
        D^2\Gamma_m[u,v]
        =\tfrac12\operatorname{tr}\!\bigl(A^{-1}DA[v]\,A^{-1}DA[u]\bigr)
         -\tfrac12\operatorname{tr}\!\bigl(A^{-1}D^2A[u,v]\bigr),
\]
and $D^3\Gamma_m$ is a finite sum of traces of products of $A^{-1}$ and $D^jA$ with $j\le3$.  Each $D^jA=D^{j+2}K_m(\vartheta_m)\,[\,\cdot\,]$ composed with derivatives of $\vartheta_m$, so $D\Gamma_m$ uses $D^3K_m$, $D^2\Gamma_m$ uses derivatives of $K_m$ up to order $4$, and $D^3\Gamma_m$ up to order $5$; by \cref{lem:tilt-control} (with the $\vartheta_m$-derivative bounds above) all three are uniformly bounded.  Expanding $\Gamma_m$ at zero and then expanding its exponential gives
\[
        \left(\frac{\det H_m}{\det\Sigma_\theta}\right)^{1/2}
        =1+m^{-1/2}D_{1,m}(z)+m^{-1}D_{2,m}(z)
          +m^{-3/2}E_{D,m}(z),
        \tag{A.1h}
\]
where $D_{1,m}$ and $D_{2,m}$ have degrees at most $1$ and $2$, and
\[
        |E_{D,m}(z)|\le C(1+\|z\|^3),
        \qquad
        \|\nabla E_{D,m}(z)\|\le C(1+\|z\|^2).
\]
Finally, $B_m(\theta_m(z))=B_m(0)+O(\|\theta_m(z)\|)$ and
$\nabla_zB_m(\theta_m(z))=O(m^{-1/2})$.  Hence the tilted local factor in \emph{(A.1f)} contributes a bounded $m^{-1}$ term, an $m^{-3/2}O(1+\|z\|)$ remainder, and a $z$-gradient of the same or smaller order.  The Fourier error $E_m^{\rm F}$ contributes $O(m^{-2})$ and, by the chain rule and \cref{lem:tilt-control}, a $z$-gradient $O(m^{-5/2})$.

To make the remainder definition explicit, put $H_m=\Sigma_{\rm ps}^-$ and define on the entire real window
\[
\begin{aligned}
        \mathcal E_m(z)
        &:=(2\pi m)^{r/2}\sqrt{\det H_m}\,
          \exp\!\left\{\frac12z^\top H_m^{-1}z\right\}\\
        &\quad\times
          \exp\!\left\{mK_m(\theta_m(z))
          -\theta_m(z)^\top(\mu^-+\sqrt m\,z)\right\}
          \mathcal J_m(z).
\end{aligned}
        \tag{A.1i}
\]
For lattice $z=(s^- -\mu^-)/\sqrt m$, equations \emph{(A.1c)} and \emph{(A.1c0)} show that $\mathcal E_m(z)$ is exactly the probability $\Pp(S_m=s)$ divided by the Gaussian prefactor in \emph{(A.1)}.

Write
\[
        a_m(z)=m^{-1/2}P_{3,m}(z)+m^{-1}P_{4,m}(z)
                +m^{-3/2}E_{5,m}(z).
\]
On $\|z\|\le A_0m^{1/8}$, $|a_m(z)|\le Cm^{-1/8}$.  Taylor's formula with integral remainder gives
\[
\begin{aligned}
        e^{a_m(z)}
        &=1+m^{-1/2}P_{3,m}(z)\\
        &\quad+m^{-1}\left\{P_{4,m}(z)
               +\frac12P_{3,m}(z)^2\right\}
          +\mathcal R_m^{\exp}(z),
\end{aligned}
\]
with
\[
        |\mathcal R_m^{\exp}(z)|
        \le Cm^{-3/2}(1+\|z\|^9),
        \qquad
        \|\nabla\mathcal R_m^{\exp}(z)\|
        \le Cm^{-3/2}(1+\|z\|^8).
\]
The derivative bound follows by applying the same integral-remainder formula to the derivative and using the bounds on $E_{5,m}$ and $\nabla E_{5,m}$.  For explicit bookkeeping, write
\[
\begin{aligned}
        F_1(z)&=1+m^{-1/2}P_{3,m}(z)
        +m^{-1}\{P_{4,m}(z)+\tfrac12P_{3,m}(z)^2\}
        +r_{1,m}(z),\\
        F_2(z)&=1+m^{-1/2}D_{1,m}(z)+m^{-1}D_{2,m}(z)+r_{2,m}(z),\\
        F_3(z)&=1+m^{-1}B_m(0)+r_{3,m}(z),
\end{aligned}
\]
where $F_1$ is the exponential factor, $F_2$ is the determinant factor, and $F_3$ is the tilted Fourier factor.  The preceding estimates give
\[
\begin{aligned}
 |r_{1,m}(z)|&\le C m^{-3/2}\{1+\|z\|^9\},
 &\|\nabla r_{1,m}(z)\|&\le C m^{-3/2}\{1+\|z\|^8\},\\
 |r_{2,m}(z)|&\le C m^{-3/2}\{1+\|z\|^3\},
 &\|\nabla r_{2,m}(z)\|&\le C m^{-3/2}\{1+\|z\|^2\},\\
 |r_{3,m}(z)|&\le C\{m^{-3/2}(1+\|z\|)+m^{-2}\},
 &\|\nabla r_{3,m}(z)\|&\le C m^{-3/2}.
\end{aligned}
\]
Here the last line uses $B_m(\theta_m(z))-B_m(0)=O(m^{-1/2}\|z\|)$ and the $C^1$ Fourier error in \emph{(A.1f)}.  Expanding $F_1F_2F_3$, define
\[
\begin{aligned}
 Q_m^{(1)}&=P_{3,m}+D_{1,m},\\
 Q_m^{(2)}&=P_{4,m}+\tfrac12P_{3,m}^2+D_{2,m}
             +P_{3,m}D_{1,m}+B_m(0).
\end{aligned}
\]
These polynomials have degrees at most $3$ and $6$ and uniformly bounded coefficients.  Every remaining monomial has an explicit factor at most $m^{-3/2}$, except terms of order $m^{-2}$ whose polynomial degree is at most $12$.  On $\|z\|\le A_0m^{1/8}$,
\[
        m^{-2}(1+\|z\|^{12})
        \le C_{A_0}m^{-3/2}(1+\|z\|^9),
\]
and, after differentiation,
\[
        m^{-2}(1+\|z\|^{11})
        \le C_{A_0}m^{-3/2}(1+\|z\|^8).
\]
The product rule together with the displayed bounds on $r_{j,m}$ therefore controls the uncollected part by $Cm^{-3/2}(1+\|z\|^9)$ and its gradient by $Cm^{-3/2}(1+\|z\|^8)$.  Define
\[
        R_m^{(2)}(z)
        :=\mathcal E_m(z)-1-m^{-1/2}Q_m^{(1)}(z)-m^{-1}Q_m^{(2)}(z).
        \tag{A.1j}
\]
The preceding bounds, together with the differentiated Fourier estimate, give
\[
        |R_m^{(2)}(z)|\le Cm^{-3/2}(1+\|z\|^9),
        \qquad
        \|\nabla R_m^{(2)}(z)\|
        \le Cm^{-3/2}(1+\|z\|^8).
\]
This proves \emph{(A.1)}--\emph{(A.2)} and also establishes that the displayed remainder is a genuine $C^1$ extension.  If $\|h\|\le Bm^{-1/2}$ and the segment from $z$ to $z+h$ stays in the window, the mean-value theorem gives
\[
\begin{aligned}
        |R_m^{(2)}(z+h)-R_m^{(2)}(z)|
        &\le \|h\|\sup_{0\le a\le1}
               \|\nabla R_m^{(2)}(z+ah)\|\\
        &\le C_Bm^{-2}(1+\|z\|^8),
\end{aligned}
\]
which is \emph{(A.3)}.  All constants are uniform because $\theta_m(z)=O(m^{-3/8})$ remains in the fixed tilted neighborhood.
\end{proof}

\begin{lemma}[Ratio lemma]\label[lemma]{lem:ratio}
Fix $A<\infty$.  Let $t\in\Z^d$ satisfy $\one^\top t=0$ and $\|t\|_1\le2$.  Let $s\in\Z_{\ge0}^d$ satisfy $\one^\top s=m$, assume $s-t\in\Z_{\ge0}^d$, and suppose
\[
        \|(s-\mu)^-\|_2\le A m^{5/8},
        \qquad z^-:=\frac{s^- -\mu^-}{\sqrt m}.
\]
Then, for all sufficiently large $m$,
\[
        \log\frac{\Pp(S_m=s-t)}{\Pp(S_m=s)}
        =\frac1m t^\top\Sigma_{\rm ps}^+(s-\mu)
        -\frac{1}{2m}t^\top\Sigma_{\rm ps}^+t
        +O_A(m^{-3/4}).
        \tag{A.4}
\]
If instead $\|(s-\mu)^-\|_2\le M\sqrt m$ for a fixed $M$, then the remainder is $O_M(m^{-1})$.
\end{lemma}

\begin{proof}
Set
\[
        z^-:=\frac{s^- -\mu^-}{\sqrt m},
        \qquad
        z'^{-}:=\frac{(s-t)^- -\mu^-}{\sqrt m}=z^- -\frac{t^-}{\sqrt m}.
\]
The hypothesis gives $\|z^-\|_2\le A m^{1/8}$.  Since $\|t\|_1\le2$, also $\|z'^{-}\|_2\le A m^{1/8}+2m^{-1/2}\le(A+1)m^{1/8}$ for large $m$.  Thus \cref{lem:edgeworth}, with $A_0=A+1$, applies at both $s$ and $s-t$.  This is the window check that is needed when the ratio is evaluated at shifted lattice points.

Taking the ratio of the two expansions cancels the normalizing determinant factor.  The Gaussian part gives
\[
-\frac12\left((z'^{-})^\top(\Sigma_{\rm ps}^-)^{-1}z'^{-}
-(z^-)^\top(\Sigma_{\rm ps}^-)^{-1}z^-\right)
=\frac1m(t^-)^\top(\Sigma_{\rm ps}^-)^{-1}(s^- -\mu^-)
-\frac{1}{2m}(t^-)^\top(\Sigma_{\rm ps}^-)^{-1}t^- .
\tag{A.5}
\]
For vectors in $\mathcal T$, reduced-coordinate bilinear forms agree with the Moore--Penrose inverse on $\mathcal T$:
\[
        (a^-)^\top(\Sigma_{\rm ps}^-)^{-1}b^-
        =a^\top\Sigma_{\rm ps}^+b,
        \qquad a,b\in\mathcal T.
        \tag{A.6}
\]
Indeed, write
\[
        \Sigma_{\rm ps}=
        \begin{pmatrix}
        B & -B\one\\
        -\one^\top B & \one^\top B\one
        \end{pmatrix},
        \qquad B=\Sigma_{\rm ps}^-,
\]
which follows from the zero row and column sums.  For $b\in\mathcal T$, the vector $x=(B^{-1}b^-,0)$ satisfies $\Sigma_{\rm ps}x=b$.  The Moore--Penrose solution $\Sigma_{\rm ps}^+b$ differs from $x$ only by an element of $\operatorname{span}\{\one\}$.  Therefore, for $a\in\mathcal T$,
\[
        a^\top\Sigma_{\rm ps}^+b
        =a^\top x
        =(a^-)^\top B^{-1}b^-,
\]
which proves (A.6).  Applying (A.6) to $a=t$ and $b=s-\mu$ converts (A.5) into the first two terms in (A.4).

It remains to bound the log-ratio of the Edgeworth factors.  Write
\[
        \mathcal E_m(z)=1+m^{-1/2}Q_m^{(1)}(z)+m^{-1}Q_m^{(2)}(z)+R_m^{(2)}(z).
\]
On the window $\|z\|_2\le(A+1)m^{1/8}$, the bounds in \cref{lem:edgeworth} give
\[
        |m^{-1/2}Q_m^{(1)}(z)|\le C m^{-1/8},
        \qquad
        |m^{-1}Q_m^{(2)}(z)|\le C m^{-1/4},
        \qquad
        |R_m^{(2)}(z)|\le C m^{-3/8}.
\]
The extension $\mathcal E_m(z)$ is real and satisfies $|\mathcal E_m(z)-1|<1/2$ for all sufficiently large $m$, uniformly on the window; in particular it is positive.  Moreover, using the derivative bounds and (A.3),
\[
\begin{aligned}
|Q_m^{(1)}(z'^{-})-Q_m^{(1)}(z^-)|
        &\le C(1+\|z^-\|_2^2)m^{-1/2},\\
|Q_m^{(2)}(z'^{-})-Q_m^{(2)}(z^-)|
        &\le C(1+\|z^-\|_2^5)m^{-1/2},\\
|R_m^{(2)}(z'^{-})-R_m^{(2)}(z^-)|
        &\le C m^{-2}(1+\|z^-\|_2^8).
\end{aligned}
\]
After multiplying by $m^{-1/2}$ and $m^{-1}$ respectively, these are $O_A(m^{-3/4})$, $O_A(m^{-7/8})$, and $O_A(m^{-1})$ on $\|z^-\|\le A m^{1/8}$.  Since $\log$ is Lipschitz on a fixed neighborhood of $1$, this gives
\[
        \log \mathcal E_m(z'^{-})-\log \mathcal E_m(z^-)=O_A(m^{-3/4}).
\]
Combining this with the Gaussian part proves (A.4).  If $\|z^-\|_2\le M$, the same estimates give $O_M(m^{-1})$ for the Edgeworth-factor contribution, and hence the improved bounded-window remainder.
\end{proof}

\begin{lemma}[Refined ratio with Edgeworth-gradient term]\label[lemma]{lem:refined-ratio}
Fix $A<\infty$ and suppose the assumptions of \cref{lem:ratio} hold with this value of $A$.  Then, for all sufficiently large $m$,
\[
        \log\frac{\Pp(S_m=s-t)}{\Pp(S_m=s)}
        =\frac1m t^\top\Sigma_{\rm ps}^+(s-\mu)
        -\frac{1}{2m}t^\top\Sigma_{\rm ps}^+t
        +\frac1m Q_{m,t}(z^-)+\rho_{m,t}(s),
        \tag{A.7}
\]
where $z^-=(s^- -\mu^-)/\sqrt m$, $Q_{m,t}$ is a quadratic polynomial satisfying
\[
        |Q_{m,t}(z^-)|\le C(1+\|z^-\|_2^2),
        \tag{A.8}
\]
and
\[
        |\rho_{m,t}(s)|\le C m^{-3/2}(1+\|z^-\|_2^9).
        \tag{A.9}
\]
All constants depend only on $(d,p,A)$.
\end{lemma}

\begin{proof}
The Gaussian quadratic part is exactly the one computed in the proof of \cref{lem:ratio}, so it gives the first two terms in (A.7).  It remains to expand the Edgeworth-factor ratio one order more carefully.

Use \cref{lem:edgeworth} with $A_0=A+1$.  As before, both $z^-$ and
\[
        z'^{-}=z^- -t^-/\sqrt m
\]
lie in the window $\|\cdot\|_2\le(A+1)m^{1/8}$ for large $m$.  Define
\[
        \mathcal E_m(z)=1+m^{-1/2}Q_m^{(1)}(z)+m^{-1}Q_m^{(2)}(z)+R_m^{(2)}(z).
\]
On this window the real-valued extension satisfies $|\mathcal E_m(z)-1|<1/2$ for all sufficiently large $m$, uniformly in $z$, and hence is positive.  Therefore
\[
        \log \mathcal E_m(z)=m^{-1/2}Q_m^{(1)}(z)+m^{-1}H_m(z)+\widetilde R_m(z),
        \tag{A.10}
\]
where
\[
        H_m(z):=Q_m^{(2)}(z)-\frac12(Q_m^{(1)}(z))^2
\]
is a polynomial of degree at most $6$ with uniformly bounded coefficients, and
\[
        |\widetilde R_m(z)|\le C m^{-3/2}(1+\|z\|_2^9).
\]
Since $\|z'^{-}-z^-\|=O(m^{-1/2})$ and both points lie in the same fixed moderate window, the pointwise bound on $\widetilde R_m$ gives
\[
        |\widetilde R_m(z'^{-})-\widetilde R_m(z^-)|
        \le |\widetilde R_m(z'^{-})|+|\widetilde R_m(z^-)|
        \le C m^{-3/2}(1+\|z^-\|_2^9).
        \tag{A.11}
\]

Subtract (A.10) at $z^-$ from (A.10) at $z'^{-}$.  The first Edgeworth polynomial gives the only order-$m^{-1}$ contribution.  Since $Q_m^{(1)}$ is cubic with uniformly bounded coefficients,
\[
\begin{aligned}
Q_m^{(1)}(z'^{-})-Q_m^{(1)}(z^-)
&=\nabla Q_m^{(1)}(z^-)\cdot(z'^{-}-z^-)
  +O\left(\|z'^{-}-z^-\|_2^2(1+\|z^-\|_2)\right)\\
&=-m^{-1/2}(t^-)^\top\nabla Q_m^{(1)}(z^-)
  +O\left(m^{-1}(1+\|z^-\|_2)\right).
\end{aligned}
\]
Multiplying by $m^{-1/2}$ yields
\[
        m^{-1/2}\{Q_m^{(1)}(z'^{-})-Q_m^{(1)}(z^-)\}
        =-\frac1m(t^-)^\top\nabla Q_m^{(1)}(z^-)
        +O\left(m^{-3/2}(1+\|z^-\|_2)\right).
\]
Define
\[
        Q_{m,t}(z):=-(t^-)^\top\nabla Q_m^{(1)}(z).
\]
Because $\nabla Q_m^{(1)}$ is quadratic with uniformly bounded coefficients, (A.8) follows.

The second-order logarithmic polynomial contributes only to the remainder.  Since $H_m$ has degree at most $6$ with uniformly bounded coefficients,
\[
        |H_m(z'^{-})-H_m(z^-)|
        \le C(1+\|z^-\|_2^5)\|z'^{-}-z^-\|_2
        \le C(1+\|z^-\|_2^5)m^{-1/2},
\]
and therefore
\[
        m^{-1}|H_m(z'^{-})-H_m(z^-)|
        \le C m^{-3/2}(1+\|z^-\|_2^5).
\]
Together with (A.11), this gives
\[
        \log \mathcal E_m(z'^{-})-\log \mathcal E_m(z^-)
        =\frac1m Q_{m,t}(z^-)
        +O\left(m^{-3/2}(1+\|z^-\|_2^9)\right).
\]
Adding the Gaussian quadratic contribution from \cref{lem:ratio} proves (A.7)--(A.9).
\end{proof}

\end{document}